\documentclass[11pt,letterpaper]{article}

\usepackage{fullpage}
\usepackage{amsthm,amsfonts} 
\usepackage{amsmath,amssymb}
\usepackage{natbib}
\usepackage{graphicx}
\usepackage[colorlinks=true, allcolors=blue]{hyperref}
\usepackage{appendix}

\usepackage{algorithm}
\usepackage{xspace}
\usepackage{algorithmic}
\usepackage{makecell}
\usepackage{multirow}
\usepackage{tabulary} 
\usepackage{diagbox}
\usepackage{todonotes}
\usepackage{cleveref}

\newtheorem{theorem}{Theorem}[section]
\newtheorem{fact}[theorem]{Fact}

\newtheorem{lemma}[theorem]{Lemma}

\newtheorem{corollary}[theorem]{Corollary}

\newtheorem{definition}[theorem]{Definition}

\newtheorem{observation}[theorem]{Observation}

\newtheorem{question}{Question}

\newcommand{\eps}{\varepsilon}
\renewcommand{\epsilon}{\varepsilon}

\newcommand{\eat}[1]{}

\newcommand{\R}{\mathbb{R}}

\newcommand{\calC}{\mathcal{C}}

\newcommand{\calX}{\mathcal{X}}

\newcommand{\OPT}{\ensuremath{\mathsf{OPT}}}

\renewcommand{\eqref}[1]{(\ref{#1})}
\newcommand{\cost}{\ensuremath{\mathrm{cost}}}

\newcommand{\ProblemName}[1]{\textsc{#1}}
\newcommand{\kzC}{\ProblemName{$(k, z)$-Clustering}\xspace}
\newcommand{\onezC}{\ProblemName{$(1, z)$-Clustering}\xspace}
\newcommand{\kMedian}{\ProblemName{$k$-Median}\xspace}
\newcommand{\oneMedian}{\ProblemName{$1$-Median}\xspace}
\newcommand{\twoMedian}{\ProblemName{$2$-Median}\xspace}

\title{On Coresets for Clustering in Small Dimensional Euclidean Spaces}
\author{
    Lingxiao Huang\footnote{  State Key Laboratory of Novel Software Technology,
    Nanjing University; Email: huanglingxiao1990@126.com}
        \and
    Ruiyuan Huang\footnote{Fudan University; Email: RuiyuanHuang00@gmail.com}
        \and
    Zengfeng Huang\footnote{Fudan University; Email: huangzf@fudan.edu.cn}
        \and
    Xuan Wu\footnote{Huawei TCS Lab; Email: wu3412790@gmail.com}   
}
\date{}

\begin{document}
\maketitle

\begin{abstract}
We consider the problem of constructing small coresets for \kMedian in Euclidean spaces. Given a large set of data points $P\subset\R^d$, a coreset is a much smaller set $S\subset\R^d$, so that the \kMedian costs of any $k$ centers w.r.t. $P$ and $S$ are close. Existing literature mainly focuses on the high-dimension case and there has been great success in obtaining dimension-independent bounds, whereas the case for small $d$ is largely unexplored. Considering many applications of Euclidean clustering algorithms are in small dimensions and the lack of systematic studies in the current literature, this paper investigates coresets for \kMedian in small dimensions. For small $d$, a natural question is whether existing near-optimal dimension-independent bounds can be significantly improved. We provide affirmative answers to this question for a range of parameters. Moreover, new lower bound results are also proved, which are the highest for small $d$. In particular, we completely settle the coreset size bound for $1$-d \kMedian (up to log factors). Interestingly, our results imply a strong separation between $1$-d \oneMedian and $1$-d \twoMedian. As far as we know, this is the first such separation between $k=1$ and $k=2$ in any dimension.
\end{abstract}

\newpage

\tableofcontents

\newpage

\section{Introduction}
\label{sec:intro}
Processing huge datasets is always computationally challenging. In this paper, we consider the coreset paradigm, which is an effective data-reduction tool to alleviate the computation burden on big data. Roughly speaking, given a large dataset, the goal is to construct a much smaller dataset, called \emph{coreset}, so that vital properties of the original dataset are preserved. 
Coresets for various problems have been extensively studied~\citep{harpeled2004on,feldman2011unified,feldman2013turning,cohenaddad2022towards,braverman2022power}.
In this paper, we investigate coreset construction for \kMedian in Euclidean spaces.

\sloppy
Coreset construction for Euclidean \kMedian has been studied for nearly two decades~\citep{harpeled2004on,feldman2011unified,huang2018epsilon,cohenaddad2021new,cohenaddad2022towards}. For this particular problem, an $\eps$-coreset is a (weighted) point set in the same Euclidean space that satisfies: given any set of $k$ centers, the \kMedian costs of the centers w.r.t.\ the original point set and the coreset are within a factor of $1+\eps$. The most important task in theoretical research here is to characterize the minimum size of $\eps$-coresets. Recently, there has been great progress in closing the gap between upper and lower bounds in high-dimensional spaces. However, researches on the coreset size in small dimensional spaces are rare. There are still large gaps between upper and lower bounds even for $1$-d \oneMedian.

Clustering in small dimensional Euclidean spaces is of both theoretical and practical importance. In practice, many applications involve clustering points in small dimensional spaces. A typical example is clustering objects in $\mathbb{R}^2$ or $\mathbb{R}^3$ based on their spatial coordinates~\citep{ SpatialClusteringExample1, SpatialClusteringExample2}. Another example is spectral clustering for graph and social network analysis~\citep{von2007tutorial, SocialNetworkExample2010, SocialNetworkExampl2014, SocialNetworkExample2017}. In spectral clustering, nodes are first embedded into a small dimensional Euclidean space using spectral methods and then Euclidean clustering algorithms are applied in the embedding space. Even the simplest $1$-d \kMedian has numerous practical applications ~\citep{ClusteringInRExample1, ClusteringInRExample2, ClusteringInRExample3}. 

On the theory side, existing techniques for coresets in high dimensions may not be sufficient to obtain optimal coresets in small dimensions. For example, much smaller size is achievable in $\mathbb{R}^1$ by using geometric methods, while the sampling methods for strong coresets in high dimension~\citep{langberg2010universal, cohenaddad2021new, Wu2022OptimalUpper} seem not viable to obtain such bounds in low dimensions. This suggests that optimal coreset construction in small dimensions may require new techniques, which provides a partial explanation of why $1$-d \oneMedian is still open after two decades of research. That being said, the coreset problem for clustering in small dimensional spaces is of great theoretical interest and practical value. Yet it is largely unexplored in the literature. This paper aims to fill the gap and study the following question:
\begin{question}
\label{que:MainQuestion}
What is the tight coreset size for Euclidean \kMedian problem in $\mathbb{R}^d$ for small $d$?
\end{question}

\subsection{Problem Definitions and Previous Results}
\paragraph{Euclidean \kMedian.} 
In the Euclidean \kMedian problem, we are given a dataset $P\subset \R^d$ ($d\geq 1$) of $n$ points and an integer $k\geq 1$; and the goal is to find a $k$-center set $C \subset \R^d$ that minimizes the objective function
\begin{equation} \label{eq:DefCost}
	\cost(P, C) := \sum_{p \in P}{d(p, C)} = \sum_{p\in P}{\min_{c\in C} d(p,c)},
\end{equation}
where $d(p,c)$ represents the Euclidean distance between $p$ and $c$.
It has many application domains including approximation algorithms, unsupervised learning, and computational geometry~\citep{lloyd1982least,tan2006cluster,arthur2007k,coates2012learning}.

\paragraph{Coresets.} 
Let $\calC$ denote the collection of all $k$-center sets, i.e., $\calC := \{ C\subset \R^d~:~ |C|=k\}$.

\begin{definition}[\bf $\eps$-Coreset for Euclidean \kMedian~\citep{harpeled2004on}]
	\label{def:coreset}
	Given a dataset $P\subset \R^d$ of $n$ points, an integer $k\geq 1$ and $\eps\in (0,1)$, an $\eps$-coreset for Euclidean \kMedian is a subset $S \subseteq P$ with weight $w : S \to \R_{\geq 0}$, such that 
	\begin{equation*}
		\forall C\in \calC,
		\qquad
		\sum_{p \in S}{w(p) \cdot d(p, C)}
		\in (1 \pm \eps) \cdot \cost(P, C).
	\end{equation*}
\end{definition}
	
\noindent
\sloppy
For Euclidean \kMedian, the best known upper bound on $\eps$-coreset size is $\Tilde{O}(\min\left\{\frac{k^{4/3}}{\eps^2}, \frac{k}{\eps^3} \right\})$~\citep{Wu2022OptimalUpper,cohenaddad2022towards} and $\Omega(\frac{k}{\eps^2})$ is the best existing lower bound~\citep{cohenaddad2022towards}. The upper bound is dimension-independent, since using dimensionality reduction techniques such as Johnson–Lindenstrauss transform, the dimension can be reduced to $\Tilde{\Theta}(\frac{1}{\eps^2})$. Thus, most previous work essentially only focus on $d=\Tilde{\Theta}(\frac{1}{\eps^2})$, whereas the case for $d < \frac{1}{\epsilon^2}$ is largely unexplored. The lower bound requires $d = \Omega(\frac{k}{\eps^2})$, as the hard instance for the lower bound is an orthonormal basis of size $\Omega(\frac{k}{\eps^2})$. For constant $k$ and large enough $d$, the upper and lower bounds match up to a polylog factor.

On the contrary, for $d \ll \Theta(\frac{1}{\eps^2})$, tight coreset sizes for \kMedian are far from well-understood, even when $k=1$. 
Specifically, for constant $d$, the current best upper bound is $\Tilde{O}(\frac{k}{\eps^3}, \frac{kd}{\eps^2})$~\citep{feldman2011unified}, and the best lower bound is $\Omega(\frac{k}{\sqrt{\eps}})$~\citep{baker2020coresets}.
%
%
Thus, there is a still large gap between the upper and lower bounds for small $d$. 
Perhaps surprisingly, this is the case even for $d=1$: \citet{HarPeled2005SmallerCF} present a coreset of size $\Tilde{O}(\frac{k}{\epsilon})$ in $\mathbb{R}$ while the best known lower bound is $\Omega(\frac{k}{\sqrt{\eps}})$.

\subsection{Our Results}
\label{sec:contribution}

We provide a complete characterization of the coreset size (up to a logarithm factor) for $d=1$ and partially answer \Cref{que:MainQuestion} for $1< d < \Theta(\frac{1}{\epsilon^2})$. Our results are summarized in Table~\ref{tab:result}.


\begin{table}
\begin{center}
    \caption{Comparison of coreset sizes for \kMedian in $\mathbb{R}^d$. We use following abbreviations: [1] for \citep{HarPeled2005SmallerCF}, [2] for \citep{feldman2011unified}, [3] for \citep{baker2020coresets}, [4] for \citep{CohenAddad2021ImprovedCA}, [5] for \citep{cohenaddad2022towards} and [6] for \citep{Wu2022OptimalUpper}. The symbol $\dagger$ represents that the results can be generalized to \kzC (\Cref{def:coreset_general}).}
\centering

\small
\begin{tabular}{|cc|c|c|c|}
\hline
\multicolumn{2}{|c|}{Paremeters $d,k$}    & Best Known Upper Bound & Best Known Lower Bound & Our Results  \\ \hline
\multicolumn{1}{|c|}{\multirow{2}{*}{$d = 1$}} & $k=1$ & $\Tilde{O}(\eps^{-1})$ [1] & $\Omega(\eps^{-1/2})$ [3] & \makecell{$\Tilde{O}(\eps^{-1/2})$ \\(Thm. ~\ref{thm:1d1k})} \\ \cline{2-5}
\multicolumn{1}{|c|}{}  & $k > 1$ & $O(k\eps^{-1})$ [1] & $\Omega(k\eps^{-1/2})$ [3]& \makecell{$\Omega(k\eps^{-1})$ \\ (Thm.~\ref{thm:Lowerbound1dkmedian})} \\ \hline
\multicolumn{1}{|c|}{\multirow{2}{*}{$1 < d < \Theta(\eps^{-2})$}} & $k=1$ & $\Tilde{O}(\eps^{-2})$ [4] & $\Omega(\eps^{-1/2})$ [3] & \makecell{$\Tilde{O}(\sqrt{d}\eps^{-1})^\dagger$ \\ (Thm.~\ref{thm:upper})}  \\ \cline{2-5} 
\multicolumn{1}{|c|}{}    & $k>1$ & $ \Tilde{O}(\min\left\{ \frac{kd}{\eps^2}, \frac{k}{\eps^3}, \frac{k^{4/3}}{\eps^2} \right\})$ [2,5, 6] & $\Omega(k\eps^{-1/2})$ [3] & \makecell{$\Omega(kd + k\eps^{-1})^\dagger$ \\ (Thm.~\ref{thm:lowerbound})} \\ \hline
%
\multicolumn{1}{|c|}{$ d = \Omega(\eps^{-2})$ }    & $k \ge 1$ & $ \Tilde{O}(\min\left\{\frac{k}{\eps^3}, \frac{k^{4/3}}{\eps^2} \right\})$ [5, 6] &$\Omega(k\eps^{-2})$ [5] & $\Biggm/$ \\ \hline
\end{tabular}
\end{center}
\label{tab:result}
\end{table}

For $d=1$, we construct coresets with size $\Tilde{O}(\frac{1}{\sqrt{\eps}})$ for \oneMedian (\Cref{thm:1d1k}) and prove that the coreset size lower bound is $\Omega(\frac{k}{\eps})$ for $k \ge 2$ (\Cref{thm:Lowerbound1dkmedian}).  
Previous work has shown coresets with size $\Tilde{O}(\frac{k}{\eps})$ exist for \kMedian~\citep{HarPeled2005SmallerCF} in $1$-d, and thus our lower bound nearly matches this upper bound. On the other hand, it was proved that the coreset size of \oneMedian in $1$-d is $\Omega(\frac{1}{\sqrt{\eps}})$~\citep{baker2020coresets}, which shows our upper bound result for \oneMedian is nearly tight. 

For $d>1$, we provide a discrepancy-based method that constructs deterministic coresets of size $\Tilde{O}(\frac{\sqrt{d}}{\eps})$ for \oneMedian (\Cref{thm:upper}). 
Our result improves over the existing $\Tilde{O}(\frac{1}{\eps^2})$ upper bound~\citep{CohenAddad2021ImprovedCA} for $1< d < \Theta(\frac{1}{\epsilon^2})$ and matches the $\Omega(\frac{1}{\eps^2})$ lower bound~\citep{cohenaddad2022towards} for $d = \Theta(\frac{1}{\epsilon^2})$. We further prove a lower bound of $\Omega(kd)$ for \kMedian in $\mathbb{R}^d$ (\Cref{thm:lowerbound}). 
Combining with our $1$-d lower bound $\Omega(\frac{k}{\eps})$, this improves over the existing $\Omega(\frac{k}{\sqrt{\eps}}+d)$ lower bound~\citep{baker2020coresets,cohenaddad2022towards}.

\subsection{Technical Overview}
\label{sec:technical}

We first discuss the 1-d \kMedian problem and show that the framework of \citep{HarPeled2005SmallerCF} is optimal with significant improvement for $k=1$.
Then we briefly summarize our approaches for $2\leq d\leq \eps^{-2}$.

\paragraph{The Bucket-Partitioning Framework for $1$-d \kMedian in \citep{HarPeled2005SmallerCF}.}
Our main results in $1$-d are based on the classic bucket-partitioning framework, developed in \citep{HarPeled2005SmallerCF}, which we briefly review now.
They greedily partition a dataset $P\subset \R$ into $O(k \eps^{-1})$ consecutive buckets $B$'s and collect the mean point $\mu(B)$ together with weight $|B|$ as their coreset $S$.
Their construction requires that the cumulative error $\delta(B) = \sum_{p\in B}|p-\mu(B)| \leq \eps\cdot \OPT/k$ holds for every bucket $B$, where $\OPT$ is the optimal \kMedian cost of $P$.
Their important geometric observation is that the induced error $|\cost(B,C) - |B|\cdot d(\mu(B),C) |$ of every bucket $B$ is at most $\delta(B)$, and even is 0 when all points in $B$ assign to the same center.
Consequently, only $O(k)$ buckets induce a non-zero error for every center set $C$ and the total induced error is at most $\eps \cdot \OPT$, which concludes that $S$ is a coreset of size $O(k \eps^{-1})$.

\paragraph{Reducing the Number of Buckets for $1$-d \oneMedian via Adaptive Cumulative Errors.}
In the case of $k=1$ where there is only one center $c\in \R$, we improve the result in~\citep{HarPeled2005SmallerCF} (\Cref{thm:1d1k}) through the following observation: $\cost(P,c)$ can be much larger than $\OPT$ when center $c$ is close to either of the endpoints of $P$, and consequently, can allow a larger induced error of coreset than $\eps \cdot \OPT$.
This observation motivates us to adaptively select cumulative errors for different buckets according to their locations.
Inspired by this motivation, our algorithm (\Cref{alg:k1d1}) first partitions dataset $P$ into blocks $B_i$ according to clustering cost, i.e., $\cost(P,c)\approx 2^i\cdot \OPT$ for all $c\in B_i$, and then further partition each block $B_i$ into buckets $B_{i,j}$ with a carefully selected cumulative error bound $\delta(B_{i,j})\leq \eps\cdot 2^i\cdot \OPT$.
Intuitively, our selection of cumulative errors is proportional to the minimum clustering cost of buckets, which results in a coreset. 
%

For the coreset size, we first observe that there are only $O(\log \eps^{-1})$ non-empty blocks $B_i$ (\Cref{lem:number_bucket}) since we can ``safely ignore'' the leftmost and the rightmost $\eps n$ points and the remaining points $p\in P$ satisfy $\cost(P,p)\leq \eps^{-1}\OPT$.
The most technical part is that we show the number $m$ of buckets in each $B_i$ is at most $O(\eps^{-1/2})$ (\Cref{lem:number_subbucket}), which results in our improved coreset size $\tilde{O}(\eps^{-1/2})$.
The basic idea is surprisingly simple: the clustering cost of a bucket is proportional to its distance to center $c$, and hence, the clustering cost of $m$ consecutive buckets is proportional to $m^2$ instead of $m$. 
According to this idea, we find that $m^2\cdot \delta(B_{i,j})\leq 2^i\cdot \OPT$ for every $B_i$, which implies a desired bound $m = O(\eps^{-1/2})$ by our selection of $\delta(B_{i,j}) \approx \eps\cdot 2^i\cdot \OPT$.

\paragraph{Hardness Result for $1$-d \twoMedian: Cumulative Error is Unavoidable.}
We take $k=2$ as an example here and show the tightness of the $O(\eps^{-1})$ bound by~\citep{harpeled2004on}.
The extension to $k > 2$ is standard via an idea of \citep{baker2020coresets}.

We construct the following worst-case instance $P\subset \R$ of size $\eps^{-1}$: We construct $m = \eps^{-1}$ consecutive buckets $B_1, B_2, \ldots, B_m$ such that the length of buckets exponentially increases while the number of points in buckets exponentially decreases.
We fix a center at the leftmost point of $P$ (assuming to be $0$ w. l. o. g.) and move the other center $c$ along the axis.
Such dataset $P$ satisfies the following: 
\begin{itemize}
\item the clustering cost is stable: for all $c$, $f_P(c) := \cost(P, \left\{0,c\right\}) \approx \eps^{-1}$ up to a constant factor;
\item the cumulative error for every bucket $B_i$ is $\delta(B_i)\approx 1$;
\item for every $B_i$, $\cost(B_i,\left\{0,c\right\})$ is a quadratic function that first decreases and then increases as $c$ moves from left to right within $B_i$, and the gap between the maximum and the minimum values is $\Omega(\delta(B_i))$.
\end{itemize}

\noindent
Suppose $S\subseteq P$ is of size $o(\eps^{-1})$.
Then there must exist a bucket $B$ such that $S\cap B = \emptyset$.
We find that function $f_S(c) := \cost(S, \left\{0,c\right\})$ is an affine linear function when $c$ is located within $B_i$ (\Cref{lem:1d2m-cost-is-affine}).
Consequently, the maximum induced error $\max_{c\in B_i} |f_P(c) - f_S(c)|$ is at least $\Omega(\delta(B_i))$ since the estimation error of an affine linear function $f_S$ to a quadratic function $f_P$ is up to certain ``cumulative curvature'' of $f_P$ (\Cref{lem:quadratic_approximation}), which is $\Omega(\delta(B_i))$ due to our construction.
Hence, $S$ is not a coreset since $f_P(c) \approx \eps^{-1}$ always holds.

We remind the readers that the above cost function $f_P$ is actually a piecewise quadratic function with $O(\eps^{-1})$ pieces instead of a quadratic one, which ensures the stability of $f_P$. 
This is the main difference from $k=1$, which leads to a gap of $\eps^{-1/2}$ on the coreset size between $k=1$ and $k=2$.
As far as we know, this is the first such separation in any dimension.

\paragraph{Our Approaches when $2\leq d\leq \eps^{-2}$.}
For \oneMedian, our upper bound result (\Cref{thm:upper}) combines a recent hierarchical decomposition coreset framework in~\citep{braverman2022power}, that reduces the instance to a hierarchical ring structure (\Cref{thm:reduction}), and the discrepancy approaches (\Cref{thm:discrepancy_upper}) developed by~\citep{karnin2019discrepancy}.
The main idea is to extend the analytic analysis of~\citep{karnin2019discrepancy} to handle multiplicative errors in a scalable way.

For \kMedian, our lower bound result (\Cref{thm:lowerbound}) extends recently developed approaches in~\citep{cohenaddad2022towards}. Their hard instance is an orthonormal basis in $\R^d$, the size of which is at most $d$, and hence cannot obtain a lower bound higher than $\Omega(d)$. We improve the results by embedding $\Theta(k)$ copies of their hard instance in $\R^d$, each of which lies in a different affine subspace. We argue that the errors from all subspaces add up. However, the error analysis from~\citep{cohenaddad2022towards} cannot be directly used; we need to overcome several technical challenges. For instance, points in the coreset are not necessary in any affine subspace, so the error in each subspace is not a corollary of their result. Moreover, errors from different subspaces may cancel each other.

\subsection{Other Related Work}
\label{sec:related}

\paragraph{Coresets for Clustering in Metric Spaces}
Recent works \citep{cohenaddad2022towards,cohen-addad2022improved,huang2023coresets} show that Euclidean \kzC admits $\epsilon$-coresets of size $\tilde{O}(k\epsilon^{-2}\cdot \min \{\epsilon^{-z},k^{\frac{z}{z+2}}\})$ and a nearly tight bound $\tilde{O}(\epsilon^{-2})$ is known when $k=1$ \citep{CohenAddad2021ImprovedCA}.
Apart from the Euclidean metric, the research community also studies coresets for clustering in general metric spaces a lot. For example, 
 \citet{feldman2011unified} construct coresets of size $\Tilde{O}(k\eps^{-2} \log n)$ for general discrete metric. \citet{baker2020coresets} show that the previous $\log n$ factor is unavoidable. There are also works on other specific metrics spaces: doubling metrics~\citep{huang2018epsilon} and graphs with shortest path metrics~\citep{baker2020coresets,BJKW21,cohenaddad2021new}, to name a few. 

\paragraph{Coresets for Variants of Clustering}
Coresets for variants of clustering problems are also of great interest. For example, \citet{braverman2022power} construct coresets of size $\Tilde{O}(k^3 \eps^{-6})$ for capacitated \kMedian, which is improved to $\Tilde{O}(k^3 \eps^{-5})$ by~\citep{huang2023coresets}. Other important variants of clustering include ordered clustering~\citep{Braverman2019CoresetsFO}, robust clustering~\citep{Robust2022}, and time-series clustering~\citep{TimeSeries2022}.


\section{Tight Coreset Sizes for $1$-d \kMedian}
\label{sec:tight}


\subsection{Near Optimal Coreset for $1$-d \oneMedian}
\label{sec:d1k1}


We have the following theorem.

\begin{theorem}[\bf Improved Coreset for one-dimensional \oneMedian]
\label{thm:1d1k}
There is a polynomial time algorithm, such that given an input data set $P\subset \mathbb{R}$, it outputs an $\eps$-coreset of $P$ for $\oneMedian$ with size $\tilde{O}(\eps^{-\frac{1}{2}})$.
\end{theorem}

\paragraph{Useful Notations and Facts.}
Throughout this section, we use $P=\{p_1,\cdots,p_n\}\subset \R$ with $p_1<p_2<\cdots<p_n$. Let $c^\star = p_{\lfloor \frac{n}{2}\rfloor}$, we have the following simple observations for $\cost(P,c)$.

\begin{observation}\label{obs:convexity}
$\cost(P,c)$ is a convex piecewise affine linear function of $c$ and $\OPT=\cost(P,c^\star)$ is the optimal \oneMedian cost on $P$.
\end{observation}

\noindent
The following notions, proposed by~\citep{harpeled2004on}, are useful for our coreset construction.

\begin{definition}[\bf Bucket]
A bucket $B$ is a continuous subset $\{p_l,p_{l+1}\dots,p_r \}$ of $P$ for some $1\leq l\leq r\leq n$. 
\end{definition}

\begin{definition}[\bf Mean and cumulative error 
\citep{HarPeled2005SmallerCF}]\label{Def:MeanAndCumulativeError}
Given a bucket $B = \{p_l,\dots,p_r \}$ for some $1\leq l\leq r\leq n$, denote $N(B) := r-l+1$ to be the number of points within $B$ and $L(B) := p_r - p_l$ to be the length of $B$. 
We define the \emph{mean} of $B$ to be  
$
\mu(B) :=\frac{1}{N(B)} \sum_{p\in B} p, 
$
and define the \emph{cumulative error} of $B$ to be
$
\delta(B) :=\sum_{p\in B} |p-\mu(B)|.
$
\end{definition}

\noindent
Note that $\mu(B)\in [p_l,p_r]$ always holds, which implies the following fact.

\begin{fact}\label{fac:del<=LN}
 $\delta(B)\leq N(B) \cdot L(B)$.
\end{fact}

\noindent
The following lemma shows that for each bucket $B$, the coreset error on $B$ is no more than $\delta(B)$.

\begin{lemma}[\bf Cumulative error controls coreset error~\citep{HarPeled2005SmallerCF}] \label{lem:CumulativeErrorControlsCoresetError}
Let $B = \left\{p_l,\ldots, p_r\right\}\subseteq P$ for $1\leq l\leq r\leq n$ be a bucket and $c\in \R$ be a center.
We have
\begin{enumerate}
\item if $c \in (p_l, p_r)$, $|\cost(B,c)-N(B) d(\mu(B),c) |\le \delta(B)$; 
\item if $c\notin (p_l, p_r)$, $|\cost(B,c)-N(B) d(\mu(B),c) | = 0$.
\end{enumerate}
\end{lemma}

\paragraph{Algorithm for \Cref{thm:1d1k}.}
Our algorithm is summarized in \Cref{alg:k1d1}.
We improve the framework in \citep{HarPeled2005SmallerCF}, which partitions $P$ into multiple buckets so that the cumulative errors in different buckets are the same and collects their means as a coreset.
Our main idea is to carefully select an adaptive cumulative error for different buckets. 
In Lines 2-3, we take the leftmost $\eps n$ points and the rightmost $\eps n$ points, and add their weighted means to our coreset $S$. 
In Lines 4 (and 7), we divide the remaining points into disjoint blocks $B_i$ ($B'_i$) such that for every $p\in B_i$, $\cost(P,p)\approx 2^i \cdot \OPT$, and then greedily divide each $B_i$ into disjoint buckets $B_{i,j}$ with a cumulative error roughly $\eps \cdot 2^i \cdot \OPT$ in Line 5. 
We remind the readers that the cumulative error in \citep{HarPeled2005SmallerCF} is always $\eps \cdot \OPT$.
%



\renewcommand{\algorithmicrequire}{\textbf{Input:}}
\renewcommand{\algorithmicensure}{\textbf{Output:}}

\begin{algorithm}
    \caption{$\mathrm{Coreset1d}(P,\eps)$}
   \label{alg:k1d1}
    \begin{algorithmic}[1]
        \REQUIRE Dataset $P=\{p_1,\cdots,p_n\}\subset \R$ with $p_1<\cdots<p_n$, and $\eps\in (0,1)$.
        \ENSURE An $\eps$-coreset $S$ of $P$ for $1$-d \oneMedian
        \STATE Set $S\leftarrow \emptyset$. 
        \STATE Set $L\leftarrow \lfloor\eps n\rfloor$ and $R\leftarrow n-\lfloor\eps n\rfloor$. Set $B_-\leftarrow \left\{p_1,\ldots, p_L\right\}$ and $B_+\leftarrow \left\{p_{R+1},\ldots, p_n\right\}$.
        \STATE \ Add $\mu(B_-)$ with weight $N(B_-)$ and $\mu(B_+)$ with weight $N(B_+)$ into $S$.
        \STATE Divide $\{p_{L+1}, \dots, p_{\lfloor \frac{n}{2} \rfloor} \}$ into disjoint blocks $\{ B_i \}_{i \ge 0}$ where $B_i := \big\{ p\in \{p_{L+1}, \dots p_{\lfloor \frac{n}{2} \rfloor} \} :  2^i\cdot \OPT\leq \cost(P,p) < 2^{i+1}\cdot \OPT \big\}$.
        \STATE For each non-empty block $B_i$ ($i\geq 0$), consider the points within $B_i$ from left to right and group them into buckets $\{ B_{i,j}\}_{j\geq 0}$ in a greedy way: each bucket $B_{i,j}$ is a maximal set with $\delta(B_{i,j})\leq \eps \cdot 2^i\cdot \OPT$.
        \STATE For every bucket $B_{i,j}$, add $\mu(B_{i,j})$ with weight $N(B_{i,j})$ into $S$.
        \STATE Symmetrically divide $\{ p_{\lfloor \frac{n}{2} \rfloor + 1},\dots, p_R\}$ into disjoint buckets $\{ B'_{i,j}\}_{i,j\geq 0}$ and add $\mu(B'_{i,j})$ with weight $N(B'_{i,j})$ into $S$ for every bucket $B'_{i,j}$.
        \STATE Return $S$.
    \end{algorithmic}
\end{algorithm}


%
\noindent
We define function $f_P: \R\rightarrow \R_{\geq 0}$ such that $f_P(c)=\cost(P,c)$ for every $c\in \R$ and define $f_S: \R\rightarrow \R_{\geq 0}$ such that $f_S(c) = \cost(S,c)$ for every $c\in \R$. By Observation~\ref{obs:convexity}, $f_P(c)$ is decreasing on $(-\infty, c^*]$ and increasing on $[c^*,\infty)$. As a result, each $B_i (B'_i)$ consists of consecutive points in $P$. 
The following lemma shows that the number of blocks 
 $B_i$($B'_i$) is $O(\log \frac{1}{\eps})$.

\begin{lemma}[\bf Number of blocks]
\label{lem:number_bucket}
There are at most $O(\log(\frac{1}{\eps}))$ non-empty blocks $B_i$ or $B'_i$.
\end{lemma}

\begin{proof}
We prove \Cref{alg:k1d1} divides $\{p_{L+1}, \dots, p_{\lfloor \frac{n}{2} \rfloor} \}$ into at most $O(\log(\frac{1}{\eps}))$ non-empty blocks $B_i$. Argument for $\{ p_{\lfloor \frac{n}{2} \rfloor + 1},\dots, p_R\}$ is entirely symmetric.

If $B_i$ is non-empty for some $i \ge 0$, we must have $f_P(p) \ge 2^i \cdot \OPT$ for $p \in B_i$. We also have $p > p_L$ since $p \in B_i \subset \{p_{L+1}, \dots, p_{\lfloor \frac{n}{2} \rfloor} \}$.  Since $f_P$ is convex, we have $ 2^i \cdot \OPT \le f_P(p) \le f_P(p_L)$. If we show that $f_P(p_L)\leq (1+\eps^{-1})\cdot \OPT = (1+\eps^{-1})\cdot f_P(c^\star)$ then we have $2^i \le  (1+\eps^{-1})$ thus $i \le O(\log(\frac{1}{\eps}))$.

To prove $f_P(p_L)\leq (1+\eps^{-1})\cdot  f_P(c^\star)$, we use triangle inequality to obtain that 
\begin{eqnarray*}
f_P(p_L)&=&\sum_{i=1}^n |p_i-p_L|\\
        &\leq& \sum_{i=1}^n (|p_i-c^\star|+|c^\star-p_L|)\\
        &=&f_P(c^\star)+n\cdot |c^\star-p_L|.
\end{eqnarray*}

Moreover, we note that by the choice of $p_L$, $|c^\star-p_L|\leq \frac{1}{L}\cdot \sum_{i=1}^L |c^\star-p_i|\leq \frac{f_P(c^\star)}{\eps n}$. Thus we have,
$$
f_P(p_L)\leq f_P(c^\star)+n\cdot \frac{f_P(c^\star)}{\eps n}=(1+\eps^{-1})\cdot f_P(c^\star).
$$
\end{proof}

\noindent
We next give a key lemma that we use to obtain an improved coreset size.

\begin{lemma}[\bf Number of buckets]
\label{lem:number_subbucket}
Each non-empty block $B_i$ or $B'_i$ is divided into $O(\eps^{-1/2})$ buckets.
\end{lemma}

\begin{proof}
We  prove that each block $B_i \subset \{p_{L+1}, \dots, p_{\lfloor \frac{n}{2} \rfloor} \}$  is divided into at most $O(\eps^{-1/2})$ buckets $B_{i,j}$. Argument for $B'_i \subset \{ p_{\lfloor \frac{n}{2} \rfloor + 1},\dots, p_R\}$ is entirely symmetric.

Suppose $B_i = \{p_{l_i},\dots,p_{r_i}\}$ and we divide $B_i$ into $t$ buckets $\{B_{i,j} \}_{j=0}^{t-1}$. Since each $B_{i,j}$ is the maximal bucket with $\delta(B_{i,j})\leq \eps \cdot 2^i\cdot \OPT$, we have $\delta(B_{i,2j} \cup B_{i,2j+1}) > \eps \cdot 2^i\cdot \OPT$ for $2j+1 < t$. Denote  $B_{i,2j} \cup B_{i,2j+1}$ by $C_j$ for $j \in \{0,\dots,\lfloor \frac{t-2}{2} \rfloor\}$, we have:

\begin{eqnarray}
4 \cdot 2^i \cdot \OPT &\geq& f_P(p_{l_i})+f_P(p_{r_i})\nonumber\\
&\geq&\sum_{p \in B_i} (|p-p_{l_i}|+|p-p_{r_i}|)\nonumber\\
&=&N(B_i)(p_{r_i}-p_{l_i})\nonumber\\
&\geq& (\sum_{j=1}^{\lfloor \frac{t-2}{2} \rfloor} N(C_j))\cdot (\sum_{j=1}^{\lfloor \frac{t-2}{2} \rfloor} L(C_j))\nonumber\\
&\geq& \big(\sum_{j=1}^{\lfloor \frac{t-2}{2} \rfloor} N(C_j)^{\frac{1}{2}}  L(C_j)^{\frac{1}{2}}\big)^2 \label{eqn:cauchy}\\
&\geq & \big(\sum_{j=1}^{\lfloor \frac{t-2}{2} \rfloor} \delta(C_j)^{\frac{1}{2}}\big)^2 \quad \textnormal{by \Cref{fac:del<=LN}}\nonumber\\
&>& (\lfloor \frac{t-2}{2} \rfloor)^2 \cdot \eps \cdot 2^i \cdot \OPT.\nonumber
\end{eqnarray}
Here \eqref{eqn:cauchy} is from Cauchy-Schwarz inequality.
So we have  $(\lfloor \frac{t-2}{2} \rfloor)^2 \cdot \eps \cdot 2^i \cdot \OPT < 4 \cdot 2^i \cdot \OPT$, which implies $t\leq O(\eps^{-\frac{1}{2}})$.
\end{proof}

\noindent
Now we are ready to prove \Cref{thm:1d1k}.

\begin{proof}[of \Cref{thm:1d1k}]
We first verify that the set $S$ is an $O(\eps)$-coreset.
Our goal is to prove that for every $c\in \R$, $f_S(c)\in (1\pm \eps)\cdot f_P(c)$.
We prove this for any $c\in (-\infty,c^\star]$. The argument for $c \in (c^\star,+\infty)$  is entirely symmetric. 

For any $c \in (-\infty,c^\star] $, we have  
$$f_P(c) - f_S(c) = \sum_{B} \cost(B,c) - N(B) \cdot d(\mu(B),c)$$ where $B$ takes over all buckets. We then separately analyze  the $c  \in (-\infty,p_L]$ case and the $c \in (p_L,c^*]$ case.
When $c  \in (-\infty,p_L]$, we note  that $f_P(p_L)=f_S(p_L)$(Lemma~\ref{lem:CumulativeErrorControlsCoresetError}). By elementary calculus, both $\frac{df_P(c)}{dc}$ and $\frac{df_S(c)}{dc}$ are within  $[-n,-(1-2\eps) n]$; hence differ by at most a multiplicative factor of $1+\eps$. Thus, $|f_P(c)-f_S(c)|\leq O(\eps)\cdot f_P(c)$. 

When $c \in (p_L,c^*]$, there is at most one bucket $B = \{ p_{l},\dots,p_r\}$ such that $c \in (p_l, p_r)$ since these buckets are disjoint. If such a bucket $B$ does not exist, we have $f_P(c) = f_S(c)$. Now suppose such a bucket $B$ exists. Since $c > p_L$, we have $B \subset B_i$ for some block $B_i$. Thus, by Lemma~\ref{lem:CumulativeErrorControlsCoresetError} and the construction of buckets:
\begin{eqnarray*}
|f_P(c) - f_S(c)|\leq  \delta(B)
\leq \eps \cdot 2^i \cdot \OPT. \\
\end{eqnarray*}
We have $f_P(p_l) \ge 2^i \cdot \OPT $ and $f_P(p_r) \ge 2^i \cdot \OPT$. Since $f_P$ is convex (thus decreasing on $(-\infty, c^*]$) and $c \in (p_l, p_r)$, we also have $f_P(c) \ge 2^i \cdot \OPT$. This implies $|f_P(c) - f_S(c)| \leq \eps \cdot f_P(c)$.

It remains to show that the size of $S$, which is the total number of buckets, is $\Tilde{O}(\eps^{-1/2})$. However, by \Cref{lem:number_bucket}, there are $O(\log(1/\eps))$ blocks, and by \Cref{lem:number_subbucket}, each block contains $O(\eps^{-1/2})$ buckets. Thus, there are at most $\Tilde{O}(\eps^{-1/2})$ buckets.
\end{proof}

\subsection{Tight Lower Bound on Coreset Size for $1$-d \kMedian when $k\geq 2$}
\label{sec:d1k2}

In this subsection, we prove that the size lower bound of $\epsilon$-coreset for \kMedian problem in $\mathbb{R}^1$ is $\Omega(\frac{k}{\epsilon})$. This lower bound matches the upper bound in~\citep{HarPeled2005SmallerCF}.

\begin{theorem}[\bf{Coreset lower bound for $1$-d \kMedian when $k\geq 2$}]\label{thm:Lowerbound1dkmedian}
For a given integer $k\geq 2$ and $\epsilon\in (0,1)$, there exists a dataset $P\subset \R$ such that any $\eps$-coreset $S$ must have size $|S|\geq \Omega(k \eps^{-1})$.
\end{theorem}

\noindent
For ease of exposition, we only prove the lower bound for \twoMedian here. 
The generalization to \kMedian is straightforward and can be found in \cref{sec:general1dkmedianlower}.

We first prove a technical lemma, which shows that a quadratic function cannot be approximated well by an affine linear function in a long enough interval. We note that similar technical lemmas appear in coresets lower bound of other related clustering problems~\citep{Braverman2019CoresetsFO}~\citep{baker2020coresets}. The lemma in~\citep{Braverman2019CoresetsFO} shows that the function $
\sqrt{x}$ cannot be approximated well by an affine linear function while our lemma is about approximating a quadratic function. The lemma in~\citep{baker2020coresets} shows that a quadratic function cannot be approximated well by an affine linear function on a bounded interval, a situation slightly different from ours.

\begin{lemma}[\bf{Quadratic function cannot be approximated well by affine linear functions}]\label{lem:quadratic_approximation}
Let $[a,b]$ be an interval, $f(c)$ be a quadratic function on interval $[a,b]$, $\alpha >0 $ and $ \beta > 0$ be two constants, and $0 \le \epsilon < \frac{1}{32} \frac{\beta}{\alpha}$ be a non-negative real number.
If $|f(c)| \le \alpha$ and $(b-a)^2 f''(c)\ge \beta$  for all $c \in [a,b]$, then there is no affine linear function $g$ such that $|g(c) - f(c)| \le \epsilon f(c)$ for all $c \in [a,b]$.
\end{lemma}

\begin{proof}
Assume there is an affine linear function $g(c)$ that satisfies $|g(c) - f(c)| \le \epsilon f(c)$ for all $c\in[a,b]$. We denote the error function by $r(c) = f(c) - g(c)$, which has two properties. First, its $l_{\infty}$ norm $\| r \|_{\infty} = \sup_{c \in [a,b]} |r(c)| \le \epsilon \alpha$. Second, it is quadratic and satisfies $r''(c) = f''(c)$, thus $(b-a)^2 r''(c) \ge \beta$  for all $c \in [a,b]$. 

Define $L = b-a$. By the mean value theorem, there is a point $c_{1/4} \in [a,\frac{a+b}{2}]$ such that $|r'(c_{1/4})| = |\frac{1}{L/2} [r(\frac{a+b}{2}) - r(a)]| \le \frac{4}{L} \| r \|_{\infty}$. Similarly there is a point $c_{3/4} \in [\frac{a+b}{2},b]$ such that $|r'(c_{3/4})| \le \frac{4}{L} \| r \|_{\infty}$. Since $r$ is a quadratic function, its derivative is monotonic and $|r'(\frac{a+b}{2})| \le \max(|r'(c_{1/4})| ,|r'(c_{3/4})|) \le \frac{4}{L} \| r \|_{\infty}$. Thus we have
\begin{align*}
    r(b) - r(\frac{a+b}{2}) &= \int_{\frac{a+b}{2}}^{b} r'(c) \mathrm{dc} \\
        &=\int_{\frac{a+b}{2}}^b r'(\frac{a+b}{2}) +  \int_{\frac{a+b}{2}}^c  r''(t) \mathrm{dt} \mathrm{dc}\\
        &= \frac{L}{2} r'(\frac{a+b}{2}) + \int_{\frac{a+b}{2}}^b \int_{\frac{a+b}{2}}^c  r''(t) \mathrm{dt} \mathrm{dc} \\
        &\ge - \frac{L}{2}  \frac{4}{L} \| r \|_{\infty} + \frac{1}{8} (b-a)^2 r''(c) \\
        &\ge -2 \epsilon \alpha + \frac{1}{8} \beta. 
\end{align*}

On the other hand $r(b) - r(\frac{a+b}{2}) \le  2 \| r \|_{\infty} \le 2 \epsilon \alpha$. We have $2\epsilon \alpha  \ge -2 \epsilon \alpha + \frac{1}{8} \beta$. Thus $\epsilon \ge \frac{1}{32} \frac{\beta}{\alpha}$.
\end{proof}

\noindent
For any dataset $P$, with a slight abuse of notations, we denote the cost function for \twoMedian with one query point fixed in $0$ by $f_P(c) = \cost(P,\{ 0, c\})$. The following lemma shows that $f_P(c)$ is a piecewise affine linear function and all the transition points are  $P\cup \{2p\mid p\in P\}$.

\begin{lemma}[\bf{The function $f_P(c)$ is piecewise affine linear}]\label{lem:1d2m-cost-is-affine}
Let $P \subset \R$ be a weighted dataset. The function $f_P(c)$ is a piecewise affine linear function. All the transition points between two affine pieces are $P\cup \{2p\mid p\in P\}$.
\end{lemma}
\begin{proof}
We denote the weight of point $p$ by $w(p)$ and denote the midpoint between any point $c$ and $0$ by $\text{mid} = \frac{c}{2}$. Now assume $c \ge 0$ and both $c$ and $\frac{c}{2}$ are not in the dataset $P$.  The clustering cost of a single point $p$ is
\[
\cost(p,\{0,c\}) = \begin{cases}
   w(p)p \quad &\text{for } p \in [0,\text{mid}], \\
   w(p)(c-p) \quad &\text{for } p \in [\text{mid},c],\\
   w(p)(p-c) \quad &\text{for } p \in [c,+\infty).
\end{cases}
\]

If $c$ changes to $c + \mathrm{dc}$ we have 
\begin{align*}
&\cost(p,\{0,c+\mathrm{dc}\}) - \cost (p,\{0,c\}) \\
= &\begin{cases}
    0 \quad &\text{for } p \in [0,\text{mid}], \\
    w(p)\mathrm{dc} \quad &\text{for } p \in [\text{mid}+\frac{1}{2} \mathrm{dc},c], \\
    -w(p)\mathrm{dc} \quad &\text{for } p \in [c+ \mathrm{dc},+\infty).
\end{cases}
\end{align*}

Assume $|\mathrm{dc}|$ is small enough, then there are no data points in $[\text{mid},\text{mid}+\frac{1}{2} \mathrm{dc}]$ and $[c,c+\mathrm{dc}]$. We have 
\begin{align*}
& \quad f_P(c+\mathrm{dc}) - f_P(c) \\
= & \quad \sum_{p \in P \cap [\text{mid},c]} w(p)  \mathrm{dc} - \sum_{p \in P \cap [c,+\infty)} w(p) \mathrm{dc},
\end{align*}
thus 
\[f_P'(c) = \sum_{p \in P \cap [\text{mid},c]} w(p) - \sum_{p \in P \cap [c,+\infty)} w(p) . \]

Consider $c$ moves in $\R$ from left to right, the derivative $f_P'(c)$ changes only when $c$ or $\text{mid} = \frac{c}{2}$ pass a data point in $P$. The same conclusion also holds for $c<0$ by a symmetric argument. This is exactly what we want.
\end{proof}

\begin{proof}[\twoMedian case of Theorem~\ref{thm:Lowerbound1dkmedian}]
We first construct the dataset $P$. The dataset $P$ is a union of $\frac{1}{\epsilon}$ disjoint intervals $\{ I_i\}_{i=1}^{\frac{1}{\epsilon}}$. Denote the left endpoint and right endpoint of $I_i$ by $l_i$ and $r_i$ respectively. We recursively define $l_i = r_{i-1}$ for $i\geq 2$,  $r_i = l_i+4^{i-1}$ for $i\geq 1$,  and $l_1 = 0$. Thus $r_i = l_{i+1} = \frac{1}{3} (4^i  - 1)$. The weight of points is specified by a measure $\lambda$ on $P$. The measure is absolutely continuous with respect to Lebesgue measure $m$ such that its density on the $i$th interval is $\frac{\mathrm{d\lambda}}{\mathrm{dm}} =  (\frac{1}{16})^{i-1}$. 
We denote the density on the $i$th interval by $\mu_i$ and the density at point $p$ by $\mu(p)$. 
Note that $P$ can be discretized in the following way. 
We only need to take a large enough constant $n$, create a bucket $B_i$ of $(\frac{1}{4})^{i-1} n $ equally spaced points in each interval $I_i$, and assign weight $\frac{1}{n}$ to every point.

The cost function $f_P(c)$ has following two features:
\begin{enumerate}
    \item \label{fea:bound} the function value $f_P(c) \in [0,\frac{2}{\epsilon}]$ for any $c \in \mathbb{R}$,
    \sloppy
    \item \label{fea:quadratic} the function is quadratic on the interval $[l_i+\frac{1}{3}(r_i - l_i), r_i]$  and satisfies $[\frac{2}{3}(r_i-l_i)]^2 f''_P(c) = \frac{2}{3}$ for each $i$.
\end{enumerate}

We show how to prove theorem~\ref{thm:Lowerbound1dkmedian} from these features and defer verification of these features later. Note that feature~\ref{fea:quadratic} does not contradict lemma~\ref{lem:1d2m-cost-is-affine} since the dataset contains infinite points.

Assume that $S$ is an $\frac{\epsilon}{300}$-coreset of $P$. We prove $|S| \ge \frac{1}{2\epsilon}$ by contradiction. If $|S| < \frac{1}{2\epsilon}$, then there is an interval $I_i = [l_i,r_i]$ such that $ (l_i,r_i) \cap S = \varnothing$ by the pigeonhole's principle. Consider function $f_S(c)$ on interval $[l_i+\frac{1}{3}(r_i - l_i), r_i]$. When $c \in [l_i+\frac{1}{3}(r_i - l_i), r_i]$, we have $\frac{c}{2} \in [l_i,r_i]$. Thus both $c$ and $\frac{c}{2}$ do not pass points in $S$ when $c$ moves from $l_i+\frac{1}{3}(r_i - l_i)$ to $r_i$. By lemma~\ref{lem:1d2m-cost-is-affine}, function $f_S(c)$ is affine linear on interval $[l_i+\frac{1}{3}(r_i - l_i), r_i]$. Since $S$ is an $\frac{\epsilon}{300}$-coreset of $P$, we have $|f_S(c) - f_P(c) | \le \frac{\epsilon}{300} f_P(c)$ on interval $[l_i+\frac{1}{3}(r_i - l_i), r_i]$. However, by applying lemma~\ref{lem:quadratic_approximation} to $f_P(c)$ and $f_S(c)$ on interval $[l_+\frac{1}{3}(r_i - l_i),r_i]$ with $\alpha = \frac{2}{\epsilon}$ and $\beta = \frac{2}{3}$, we obtain that $\frac{\epsilon}{300} \ge \frac{1}{32} \times \frac{2}{3} \times \frac{\epsilon}{2} > \frac{\epsilon}{300}$. This is a contradiction.

It remains to verify the two features of $f_P(c)$. We verify feature~\ref{fea:bound} by direct computations. For any point $c$, the function satisfies 
\begin{align*}
    0 \le f_P(c) &\le \cost(P,\{0,0\})  = \int_P p \mu(p) \mathrm{dp} \\
    &\le \sum_{i=1}^{\frac{1}{\epsilon}} \lambda(I_i) r_i \le  \sum_{i=1}^{\frac{1}{\epsilon}} (\frac{1}{4})^{i-1} \times 2 \times 4^{i-1}\\
    &=\frac{2}{\epsilon}.
\end{align*}

To verify feature~\ref{fea:quadratic}, we compute the first order derivative by computing the change of the function value $f_P(c+\mathrm{dc}) - f_P(c)$ up to the first order term when $c$ increases an infinitesimal number $\mathrm{dc}$. The unweighted clustering cost of a single point $p$ is
\[
\cost(p,\{0,c\}) = \begin{cases}
   p \quad &\text{for } p \in [0,\text{mid}], \\
   c-p \quad &\text{for } p \in [\text{mid},c],\\
   p-c \quad &\text{for } p \in [c,+\infty).
\end{cases}
\]

As $c$ increases to $c + \mathrm{dc}$, the clustering cost of a single point changes
\begin{align*}
& \quad \cost(p,\{0,c+\mathrm{dc}\}) - \cost (p,\{0,c\})] \\
= & \quad \begin{cases}
    0 \quad &\text{for } p \in [0,\text{mid}], \\
    O(\mathrm{dc}) \quad &\text{for } p \in [\text{mid},\text{mid}+\frac{1}{2} \mathrm{dc}], \\
    \mathrm{dc} \quad &\text{for } p \in [\text{mid}+\frac{1}{2} \mathrm{dc},c], \\
    O(\mathrm{dc}) \quad &\text{for } p \in [c,c + \mathrm{dc}], \\
    -\mathrm{dc} \quad &\text{for } p \in [c+ \mathrm{dc},+\infty).
\end{cases}
\end{align*}

The cumulative clustering cost changes
\begin{align*}
    & f_P(c +\mathrm{dc}) - f_P(c) \\
    = & \int_0^{+\infty} \cost(p,\{0,c+\mathrm{dc}\}) - \cost (p,\{0,c\}) \mathrm{d\lambda} \\
    = & \int_0^{\text{mid}} 0\mathrm{d\lambda} + \int_{\text{mid}}^{\text{mid}+\frac{1}{2} \mathrm{dc}} O(\mathrm{dc}) \mathrm{d\lambda} + \int_{\text{mid}+\frac{1}{2} \mathrm{dc}}^c \mathrm{dc} \mathrm{d\lambda}\\
    &+\int_c^{c+ \mathrm{dc}} O(\mathrm{dc}) \mathrm{d\lambda} + \int_{c+ \mathrm{dc}}^{+\infty} -\mathrm{dc} \mathrm{d\lambda}\\
    = & \lambda([\text{mid},c])\mathrm{dc} - \lambda([c,+\infty))\mathrm{dc} + O(\mathrm{dc})^2.
\end{align*}

Thus the first order derivative $f_P'(c) = \lambda([\frac{c}{2},c]) - \lambda([c,+\infty))$ and the second order derivative
\begin{align*}
    f_P''(c) &= \frac{\mathrm{d}}{\mathrm{dc}} \bigl( \lambda([\frac{c}{2},c]) - \lambda([c,+\infty)) \bigr),\\
 &= 2\mu(c) - \frac{1}{2} \mu(\frac{c}{2}) .
\end{align*} 

For $c \in [l_i+\frac{1}{3} (r_i - l_i), r_i]$, the two points $c$ and $\frac{c}{2}$ both lie in 
interval $[l_i, r_i]$. We have $\mu(c) = \mu(\frac{c}{2}) = \mu_i$ and $f_P''(c) = \frac{3}{2} \mu_i$. Thus the function $f_P(c)$ is quadratic on $[l_i+\frac{1}{3}(r_i - l_i), r_i]$ and $[\frac{2}{3}(r_i-l_i)]^2 f''_P(c) = \frac{2}{3}$.
\end{proof}
\section{Improve Coreset Sizes when $2\leq d \leq \eps^{-2}$}
\label{sec:improve}

In this section, we consider the case of constant $d$, $2\leq d\leq \eps^{-2}$, and provide several improved coreset bounds for a general problem of Euclidean \kMedian, called Euclidean \kzC.
The only difference from \kMedian is that the goal is to find a $k$-center set $C \subset \R^d$ that minimizes the objective function
\begin{equation} \label{eq:DefCost_general}
    \cost_z(P, C) := \sum_{p \in P}{d^z(p, C)} = \sum_{p\in P}{\min_{c\in C} d^z(p,c)},
\end{equation}
where $d^z$ represents the $z$-th power of the Euclidean distance.
The coreset notion is as follows.

\begin{definition}[\bf $\epsilon$-Coreset for Euclidean \kzC~\citep{harpeled2004on}]
	\label{def:coreset_general}
	Given a dataset $P\subset \R^d$ of $n$ points, an integer $k\geq 1$, constant $z\geq 1$ and $\eps\in (0,1)$, an $\eps$-coreset for Euclidean \kzC is a subset $S \subseteq P$ with weight $w : S \to \R_{\geq 0}$, such that 
	\begin{equation*}
		\forall C\in \calC,
		\qquad
		\sum_{p \in S}{w(p) \cdot d^z(p, C)}
		\in (1 \pm \eps) \cdot \cost_z(P, C).
	\end{equation*}
\end{definition}
	
\noindent
We first study the case of $k=1$ and provide a coreset upper bound $\tilde{O}(\sqrt{d} \eps^{-1})$ (\Cref{thm:upper}).
Then we study the general case $k\geq 1$ and provide a coreset lower bound $\Omega(kd)$ (\Cref{thm:lowerbound}).

\subsection{Improved Coreset Size in $\R^d$ when $k=1$}
\label{sec:d2k1}

We prove the following main theorem for $k=1$ whose center is a point $c\in \R^d$.

\begin{theorem}[\bf{Coreset for Euclidean \onezC}]
    \label{thm:upper}
    Let integer $d\geq 1$, constant $z\geq 1$ and $\eps\in (0,1)$.
    There exists a randomized polynomial time algorithm that given a dataset $P\subset \R^d$, outputs an $\eps$-coreset for Euclidean \onezC of size at most $z^{O(z)} \sqrt{d}\eps^{-1}\log \eps^{-1}$.
\end{theorem}

\begin{proofsketch}
    By~\citep{braverman2022power}, we first reduce the problem to constructing a mixed coreset $(S,w)$ for Euclidean \onezC for a dataset $P\subset B(0,1)$ satisfying that $\forall c\in \R^d$,
    \[
	\sum_{p \in S}{w(p) \cdot d^z(p, c)} \in
	\cost_z(P, c)\pm \eps \max\left\{1, \|c\|_2\right\}^z \cdot |P|.
    \]
    The main idea to construct such $S$ is to prove that the class discrepancy of Euclidean \onezC for $P$ is at most $z^{O(z)}\max\left\{1,r\right\}^z \cdot \sqrt{d}/m$ for $c\in B(0,r)$ (\Cref{lm:discrepancy_median}), which implies the existence of a mixed coreseet $S$ of size $z^{O(z)} \sqrt{d}\eps^{-1}$ by Fact 6 of~\citep{karnin2019discrepancy}.
    For the class discrepancy, we apply an analytic result of ~\citep{karnin2019discrepancy} (\Cref{thm:discrepancy_upper}).
    The main difference is that \citep{karnin2019discrepancy} only considers an additive error that can handle $c\in B(0,1)$ instead of an arbitrary center $c\in \R^d$.
    In our case, we allow a mixed error proportional to the scale of $\|c\|_2$ and extend the approach of \citep{karnin2019discrepancy} to handle arbitrary centers $c\in \R^d$ by increasing the discrepancy by a multiplicative factor $\|c\|_2^z$.
\end{proofsketch}

\noindent
The above theorem is powerful and leads to the following results for $z=O(1)$:
\begin{enumerate}
    \item By dimension reduction as in~\citep{huang2020coresets,cohenaddad2021new,cohenaddad2022towards}, we can assume $d=O(\eps^{-2}\log\eps^{-1})$.
    Consequently, our coreset size is upper bounded by $\tilde{O}(\eps^{-2})$, which matches the nearly tight bound in~\citep{cohenaddad2022towards}.
    \item For $d= O(1)$, our coreset size is $O(\eps^{-1})$, which is the first known result in small dimensional space.
    Specifically, the prior known coreset size in $\R^2$ is $\tilde{O}(\eps^{-3/2})$~\citep{braverman2022power}, and our result improves it by a factor of $\eps^{-1/2}$.
\end{enumerate}

\noindent
We conjecture that our coreset size is almost tight, i.e., there exists a coreset lower bound $\Omega(\sqrt{d} \eps^{-1})$ for constant $2\leq d \leq \eps^{-2}$, which leaves as an interesting open problem.

\subsubsection{Useful Notations and Facts}
\label{sec:notation_upper}

For preparation, we first propose a notion of mixed coreset (\Cref{def:mixed_coreset}), and then introduce some known discrepancy results.

\paragraph{Reduction to mixed coreset.}
Let $B(a,r)$ denote the $\ell_2$-ball in $\R^d$ that centers at $a\in \R^d$ with radius $r\geq 0$.
Specifically, $B(0,1)$ is the unit ball centered at the original point.

\begin{definition}[\bf Mixed coreset for Euclidean \onezC]
	\label{def:mixed_coreset}
	Given a dataset $P\subset B(0,1)$ and $\eps\in (0,1)$, an $\eps$-mixed-coreset for Euclidean \onezC is a subset $S \subseteq P$ with weight $w : S \to \R_{\geq 0}$, such that $\forall c\in \R^d$,
	\begin{equation} 
	\label{eq:DefCoreset}
	\sum_{p \in S}{w(p) \cdot d^z(p, c)} \in
	\cost_z(P, c)\pm \eps \max\left\{1, \|c\|_2\right\}^z \cdot |P|.
	\end{equation}
\end{definition}

\noindent
Actually, prior work~\citep{cohenaddad2021new,cohenaddad2022towards,braverman2022power} usually consider the following form: $
	\forall c\in \R^d$,
\begin{equation} 
	\label{eq:DefCoreset_equivalent}
	\sum_{p \in S}{w(p) \cdot d^z(p, c)}
	\in (1 \pm \eps) \cdot \cost_z(P, c)\pm \eps |P|.
\end{equation}

Compared to \Cref{def:coreset}, the above inequality allows both a multiplicative error $\eps\cdot \cost_z(P,c)$ and an additional additive error $\eps |P|$.
Note that for a small $r=O(1)$, the additive error $\eps |P|$ dominates the total error; while for a large $r\gg \Omega(1)$, the multiplicative error $\eps\cdot \cost_z(P,c)\approx \eps \|c\|_2\cdot |P|$ dominates the total error.
Hence, it is not hard to check that Inequality~\eqref{eq:DefCoreset_equivalent} is an equivalent form of Inequality~\eqref{eq:DefCoreset} (up to an $2^{O(z)}$-scale).
This is also the reason that we call \Cref{def:mixed_coreset} mixed coreset.
We have the following useful reduction.

\begin{theorem}[\bf{Reduction from coreset to mixed coreset~\citep{braverman2022power}}]
    \label{thm:reduction}
    Let $\eps\in (0,1)$.
    Suppose there exists a polynomial time algorithm $A$ that constructs an $\eps$-mixed coreset for Euclidean \onezC of size $\Gamma$.
    Then there exists a polynomial time algorithm $A'$ that constructs an $\eps$-coreset for Euclidean \onezC of size $O(\Gamma\log \eps^{-1})$.
\end{theorem}

\noindent
Thus, it suffices to prove that an $\eps$-mixed coreset is of size $z^{O(z)}\sqrt{d} \eps^{-1}$, which implies \Cref{thm:upper}.

\paragraph{Class discrepancy.}
For preparation, we introduce the notion of class discrepancy introduced by~\citep{karnin2019discrepancy}.
The idea of combining discrepancy and coreset construction has been studied in the literature, specifically for kernel density estimation~\citep{phillips2018improved,phillips2018near,karnin2019discrepancy,tai2022optimal}.
We propose the following definition.

        \begin{definition}[\bf{Class discrepancy~\citep{karnin2019discrepancy}}]
	    \label{def:discrepancy}
	    Let $m\geq 1$ be an integer.
	    Let $f: \calX, \calC\rightarrow \R$ and $P\subseteq \calX$ with $|P|=m$.
	    The class discrepancy of of $P$ w.r.t. $(f,\calC)$ is 
	    \begin{align*}
	    D^{(\calC)}_P(f) := & \min_{\sigma\in \left\{-1,1\right\}^P} D^{(\calC)}_P(f,\sigma) \\
        = &\min_{\sigma\in \left\{-1,1\right\}^P} \max_{c\in \calC} \frac{1}{m}\left|\sum_{p\in P} \sigma_p\cdot f(p,c)\right|.
	    \end{align*}
	    Moreover, we define $D^{(\calX, \calC)}_m(f) := \max_{P\subseteq \calX: |P|=m} D^{(\calC)}_P(f)$ to be the class discrepancy w.r.t. $(f,\calX,\calC)$.
	\end{definition}
	
	\noindent
	Here, $\calX$ is the instance space and $\calC$ is the parameter space.
	Specifically, for Euclidean \onezC, we let $\calX, \calC\subseteq \R^d$ and $f$ be the Euclidean distance.
	The class discrepancy $D^{(\calX, \calC)}_m(f)$ measures the capacity of $\calC$.
	Intuitively, if the capacity of $\calC$ is large and leads to a complicated geometric structure of vector $(f(p,c))_{p\in P}$ for $c\in \calC$, $D^{(\calX, \calC)}_m(f)$ tends to be large.

\paragraph{Useful discrepancy results.}
For a vector $p\in \R^d$ and integer $l\geq 1$, let $p^{\otimes l}$ present the $l$-dimensional tensor obtained from the outer product of $p$ with itself $l$ times.
For a $l$-dimensional tensor $X$ with $d^l$ entries, we consider the measure $\|X\|_{T_l} := \max_{c\in \R^d: \|q\| = 1} |\langle X, q^{\otimes l}\rangle |$.
Next, we provide some known results about the class discrepancy.

\begin{theorem}[\bf{An upper bound for class discrepancy (restatement of Theorem 18 of~\citep{karnin2019discrepancy})}]
    \label{thm:discrepancy_upper}
    \sloppy
    Let $\calX = B(0,1)$ in $\R^d$.
    Let $f:\R\rightarrow \R$ be analytic satisfying that for any integer $l\geq 1$, $|\frac{d^l f}{d x^l}(x)| \leq \gamma_1 C^l l!$ for some constant $\gamma_1, C>0$.
    Let $\calC = B(0,\frac{1}{2C})$ and $m\geq 1$ be an integer.
    The class discrepancy w.r.t. $(f = f(\langle p,c \rangle), \calX, \calC)$ is at most $D^{(\calX, \calC)}_m(f) \leq \gamma_2\gamma_1\sqrt{d}/m$ for some constant $\gamma_2 > 0$.

     Moreover, for any dataset $P\subset \calX$ of size $m$, there exists a randomized polynomial time algorithm that constructs $\sigma\in \left\{-1,1\right\}^P$ satisfying that for any integer $l\geq 1$, we have
    \begin{align*}
    \label{ineq:tensor}
    \|\sum_{p\in P}\sigma_p\cdot p^{\otimes l}\|_{T_l} = O(\sqrt{d l \log^3 l}).
    \end{align*}
    This $\sigma$ satisfies $D^{(\calC)}_P(f,\sigma) \leq \gamma_2\gamma_1\sqrt{d}/m$.
\end{theorem}

\noindent
Note that the above theorem is a constructive result instead of an existential result in Theorem 18 of~\citep{karnin2019discrepancy}.
This is because Theorem 18 of~\citep{karnin2019discrepancy} applies the existential version of Banaszczyk’s \citep{banaszczyk1998balancing}, which has been proven to be constructive recently~\citep{bansal2019gram}.
Also, note that the construction of $\sigma$ only depends on $P$ and does not depend on the selection of $\calC$.
This observation is important for the construction of mixed coresets via discrepancy.

\subsubsection{Proof of \Cref{thm:upper}}
\label{sec:proof_upper}

We are ready to prove \Cref{thm:upper}.
The main lemma is as follows.

\begin{lemma}[\bf Class discrepancy for Euclidean \onezC]
    \label{lm:discrepancy_median}
    Let $m\geq 1$ be an integer.
    Let $f = d^z$ and $\calX = B(0,1)$.
    For a given dataset $P\subset \calX$ of size $m$, there exists a vector $\sigma\in \left\{-1,1\right\}^P$ such that for any $r>0$,
    \[ 
    D^{(B(0,r))}_P(f,\sigma) \leq z^{O(z)}\max\left\{1,r\right\}^z \cdot \sqrt{d}/m.
    \]
\end{lemma}

\noindent
The above lemma indicates that the class discrepancy for Euclidean \onezC linearly depends on the radius $r$ of the parameter space $\calC$.
Note that the lemma finds a vector $\sigma$ that satisfies all levels of parameter spaces $\calC= B(0,r)$ simultaneously. 
This requirement is slightly different from \Cref{def:discrepancy} that considers a fixed parameter space.
Observe that the term $\max\left\{1,r\right\}$ is similar to $\max\left\{1,\|c\|_2\right\}$ in \Cref{def:mixed_coreset}, which is the key of reduction from \Cref{lm:discrepancy_median} to \Cref{thm:upper}.
The proof idea is similar to that of Fact 6 of~\citep{karnin2019discrepancy}.

\begin{proof}[of \Cref{thm:upper}]
    Let $P\subset B(0,1)$ be a dataset of size $n$ and $\Lambda = z^{O(z)} \sqrt{d}\eps^{-1}$.
    By the same argument as in Fact 6 of~\citep{karnin2019discrepancy}, we can iteratively applying \Cref{lm:discrepancy_median} to construct a subset $S\subseteq P$ of size $m = \Theta(\Lambda)$ together with weights $w(p) = \frac{n}{|S|}$ for $p\in S$ and a vector $\sigma\in \left\{-1,1\right\}^S$, and $(S,\sigma)$ satisfies that for any $c\in \R^d$,
    \begin{align*}
    & \quad \left| \sum_{p\in S} w(p)\cdot d(p,c) - \cost_z(P,c) \right| \\
    \leq & \quad n\cdot D^{(B(0,\|c\|_2))}_S(f,\sigma) \\
    \leq & \quad \eps \max\left\{1, \|c\|_2\right\}\cdot n.
    \end{align*}
    %
    %
    
    This implies that $S$ is an $O(\eps)$-mixed coreset for Euclidean \onezC of size at most $\Lambda = z^{O(z)}\sqrt{d} \eps^{-1}$, which completes the proof of \Cref{thm:upper}.
\end{proof}

\noindent
It remains to prove \Cref{lm:discrepancy_median}.

\begin{proof}[of \Cref{lm:discrepancy_median}]
    Let $P\subset B(0,1)$ be a dataset of size $m$.
    We first construct a vector $\sigma\in \left\{-1,1\right\}^P$ by the following way: 
    \begin{enumerate}
        \item For each $p\in P$, construct a point $\phi(p)=(\frac{1}{2}\|p\|_2^2, \frac{\sqrt{2}}{2} p, \frac{1}{2})\in \R^{d+2}$.
        \item By \Cref{thm:discrepancy_upper}, construct $\sigma\in \left\{-1,1\right\}^P$ such that for any integer $l\geq 1$, 
        \[
        \|\sum_{p\in P}\sigma_p\cdot \phi(p)^{\otimes l}\|_{T_l} = O(\sqrt{(d+2) l \log^3 l}).
        \]
    \end{enumerate}
    Let $\phi(P)$ be the collection of all $\phi(p)$s.
    Note that $\|\phi(p)\|_2\leq 1$ by construction, which implies that $\phi(P)\subset B(0,1)\subset \R^{d+2}$.
    In the following, we show that $\sigma$ satisfies \Cref{lm:discrepancy_median}.

    Fix $r\geq 1$ and let $\calC=B(0,r)$.
    We construct another dataset $P' = \left\{p'=\frac{p}{4r}: p\in P\right\}$.
    For any $c\in \calC=B(0,r)$, we let $c' = \frac{c}{4r}\in B(0,\frac{1}{4})$.
    By definition, we have for any $p\in \calX$ and $c\in \calC$,
    \[
    \frac{1}{m} \left|\sum_{p\in P}\sigma_p\cdot f(p,c)\right| = \frac{(4r)^z}{m} \left|\sum_{p'\in P'}\sigma_p\cdot f(p',c')\right|,
    \]
    which implies that 
    \[
    D^{(\calC)}_P(f,\sigma) = (4r)^z\cdot D^{(B(0,\frac{1}{4}))}_{P'}(f,\sigma).
    \]
    Thus, it suffices to prove that 
    \begin{align}
    \label{ineq1_proof_lm:discrepancy_median}
    D^{(B(0,\frac{1}{4}))}_{P'}(f,\sigma) \leq z^{O(z)}\sqrt{d}/m,
    \end{align}
    which implies the lemma.
    The proof idea of Inequality~\eqref{ineq1_proof_lm:discrepancy_median} is similar to that of Theorem 22 of~\citep{karnin2019discrepancy}.\footnote{Note that the proof of Theorem 22 of~\citep{karnin2019discrepancy} is actually incorrect. Applying Theorem 18 of~\citep{karnin2019discrepancy} may lead to an upper bound $\|\tilde{q}\|_2< 1$, which makes $R$ in Theorem 22 of~\citep{karnin2019discrepancy} not exist.}
    For each $p'\in P'$ and $c'\in B(0,\frac{1}{4})$, let $\psi(c') = (\frac{1}{8r^2}, -\frac{\sqrt{2}}{2r} c', 2\|c'\|_2^2)\in \R^{d+2}$ and we can rewrite $f(p',c')$ as follows:
    \[
    f(p',c') = \|p'-c'\|_2^z = (\left\langle \phi(p), \psi(c') \right\rangle)^{z/2}.
    \]
    We note that $\phi(p)\in B(0,1)$ and $\psi(c')\in B(0,\frac{1}{3})$ since $c'\in B(0,\frac{1}{4})$.
    Construct another function $g: P\times B(0,\frac{1}{3})$ as follows: for each $p\in P$ and $c\in B(0,\frac{1}{3})$,
    \begin{enumerate}
        \item If for any $p'\in P$, $\langle p',c\rangle \geq 0$, let $g(p,c) = g(\langle p,c\rangle) = (\langle p,c\rangle)^{z/2}$;
        \item Otherwise, let $g(p,c) = 0$.
    \end{enumerate}
    We have $|\frac{d^l g}{d x^l}(x)| \leq z^{O(z)} l!$ for any integer $l\geq 1$.
    By the construction of $\sigma$ and \Cref{thm:discrepancy_upper}, we have that 
    \[
    D^{(B(0,\frac{1}{3}))}_{\phi(P)}(g,\sigma) \leq z^{O(z)}\sqrt{d}/m,
    \]
    which implies Inequality~\eqref{ineq1_proof_lm:discrepancy_median} since $D^{(B(0,\frac{1}{4}))}_{P'}(f,\sigma)\leq D^{(B(0,\frac{1}{3}))}_{\phi(P)}(g,\sigma)$ due to the fact that $\psi(c')\in B(0,\frac{1}{3})$.

    Overall, we complete the proof.
\end{proof}

\subsection{Improved Coreset Lower Bound in $\R^d$ when $k\geq 2$}\label{sec:lb_main}

We present a lower bound for the coreset size in small dimensional spaces. 

\begin{theorem}[\bf{Coreset lower bound in small dimensional spaces}]
\label{thm:lowerbound}
Given an integer $k\geq 1$, constant $z\geq 1$ and a real number $\epsilon \in (0,1)$, for any integer $d \le \frac{1}{100\epsilon^2}$, there is a dataset $P \subset \R^{d+1}$ such that its $\epsilon$-coreset for \kzC must contain at least $\frac{dk}{10 z^4}$ points.
\end{theorem}

\noindent
When $d = \Theta(\frac{1}{\epsilon^2})$, \Cref{thm:lowerbound} gives the well known lower bound $\frac{k}{\epsilon^2}$. When $d \ll \Theta(\frac{1}{\epsilon^2})$, the theorem is non-trivial. 
In the following, we prove \Cref{thm:lowerbound} for $z = 2$ and show how to extend to general $z \ge 1$ in \Cref{sec:generalz}.

\subsubsection{Preparation}
\paragraph{Notations} Let $e_0, \cdots, e_d$ be the standard basis vectors of $\R^{d+1}$, and $H_1,\cdots,H_{k/2}$ be $k/2$ $d$-dimensional affine subspaces, where $H_j := jLe_0 + \text{span}\left\{e_1,\dots,e_d\right\}$ for a sufficiently large constant $L$. For any $p\in \R^{d+1}$, we use $\Tilde{p}$ to denote the $d$-dimensional vector $p_{1:d}$ (i.e., discard the $0$-th coordinate of $p$).

\paragraph{Hard instance} 
%
We construct the hard instance as follows. Take $P_j = \{ jLe_0+e_1,\cdots, jLe_0+e_{d/2} \}$ for $j \in \{1, \dots, k/2\}$ and take $P$ to be the union of all $P_j$. 
The hard instance is $P$.
%
%
Note that $P_j \subset H_j$ for each $j$ and $|P|=kd/4$.
\color{black}
In our proof, we always put two centers in each $H_j$. Thus for large enough $L$, all $p\in P_j$ must be assigned to centers in $H_j$.

We will use the following two technical lemmas from \citep{cohenaddad2022towards}.
\begin{lemma}\label{lem:cost-to-basis}
For any $k\ge 1$, let $\{c_1,\cdots,c_k\}$ be arbitrary $k$ unit vectors in $\R^d$, we have
\begin{align*}
    \sum_{i=1}^{d/2} \min_{\ell=1}^k\|e_i-c_{\ell}\|^2 \ge d-\sqrt{dk/2}.
\end{align*}
\end{lemma}

\begin{lemma}\label{lem:cost-to-smallset}
Let $S$ be a set of points in $\R^d$ of size $t$ and $w: S\rightarrow \R^+$ be their weights. There exist $2$ unit vectors $v_1, v_2$, such that
\begin{align*}
    \sum_{p\in S} w(p)\min_{\ell=1,2} \|p-v_{\ell}\|^2 \le \sum_{s\in P}w(p)(\|p\|^2+1)-  \frac{2\sum_{p\in S} w(p)\|p\|}{\sqrt{t}}.
\end{align*}
\end{lemma}

\subsubsection{Proof of \Cref{thm:lowerbound} when $z=2$}

\noindent
Now we are ready to prove \Cref{thm:lowerbound} when $z=2$. 

\begin{proof}
Note that points in $S$ might not be in any $H_j$. We first map each point $p\in S$ to an index $j_p\in [k/2]$ such that $H_{j_p}$ is the nearest subspace of $p$. The mapping is quite simple: 
\begin{align*}
    j_p = \arg\min_{j\in[k/2]} |p_0 - jL|,
\end{align*}
where $p_0$ is the $0$-th coordinate of $p$.
Let $\Delta_p = p_0 - j_pL$, which is the distance of $p$ to the closest affine subspace. Let $S_j:=\{p\in S: j_p=j\}$ be the set of points in $P$, whose closest affine subspace is $H_j$. Define $I:=\{j\in[k/2] : |S_j| \le d/4\}$. Consider any $k$-center set $C$ such that $H_j\bigcap C \neq \emptyset$. Then $\cost(P, C)\ll 0.1 L$ for sufficiently large $L$. On the other hand, $\cost(S, C) \ge \sum_{p\in S} \Delta_p^2$. Since $S$ is a coreset, $\Delta_p^2 \ll L$ for all $p\in S$. \footnote{Here we do not allow offsets to simplify the proof, but our technique can be extended to handle offsets.} Therefore each $p\in S$ must be very close to its closest affine subspace; in particular, we can assume that $p$ must be assigned to some center in $H_{j_p}$ (if there exists one).

In the proof follows, we consider three different set of $k$ centers $C_1,C_2$ and $C_3$ and compare the costs $\cost(P, C_i)$ and $\cost(S, C_i)$ for $i=1,2,3$. In each $C_i$, there are two centers in each $H_j$. As we have discussed above, for large enough $L$, the total cost for both $P$ and $S$ can be decomposed into the sum of costs over all affine subspaces.

For each $j\in \Bar{I}$, the corresponding centers in $H_j$ are the same across $C_1,C_2,C_3$. Let $c_j$ be any point in $H_j$ such that $c_j-jLe_0$ has unit norm and is orthogonal to $e_1,\cdots,e_{d/2}$; in other words, $\|\Tilde{c}_j\|=1$ and the first $d/2$ coordinates of $\Tilde{c}_j=1$ are all zero. Specifically, we set $c_j = jLe_0 + e_{d/2+1}$ and the two centers in $H_j$ are two copies of $c_j$ for $j\in \Bar{I}$. 

We first consider the following $k$ centers denoted by $C_1$. As we have specified the centers for $j\in \Bar{I}$, we only describe the centers for each $j\in I$. Since by definition, $|S_j|\le d/4$, we can find a vector $c_j\in \R^{d+1}$ in $H_j$ such that $c_j-jLe_0$ has unit norm and is orthogonal to $e_1,\cdots,e_{d/2}$ and all vectors in $S_j$. Let $C_1$ be the set of $k$ points with each point in $\{c_1,\cdots,c_{k/2}\}$ copied twice. We evaluate the cost of $C_1$ with respect to $P$ and $S$. 
\begin{lemma}
    For $C_1$ constructed above, we have $\cost(P,C_1) = \frac{kd}{2}$ and 
    \begin{align*}
        \cost(S,C_1) =\sum_{p\in S} w(p)(\Delta_p^2+\|\Tilde{p}\|^2 +1) - 2\sum_{j\in \Bar{I}}\sum_{p\in S_j}w(p)\langle p-jLe_0, jLe_0-c_j \rangle.
    \end{align*}
\end{lemma}
\begin{proof}
Since $e_i$ is orthogonal to $c_j-jLe_0$ and $c_j-jLe_0$ has unit norm for all $i,j$, it follows that
\begin{align}
\cost(P,C_1) &=
    \sum_{j=1}^{k/2}\sum_{i=1}^{d/2} \min_{c\in C_1}\|jLe_0+e_i-c\|^2 = \sum_{j=1}^{k/2} \sum_{i=1}^{d/2} \|jLe_0+e_i -c_j\|^2 \nonumber\\
    &=\sum_{j=1}^{k/2}\sum_{i=1}^{d/2}(\|e_i\|^2+\|c_j-jLe_0\|^2-2\langle e_i, c_j-jLe_0 \rangle )\nonumber\\
    &=\frac{kd}{2}.\label{eqn:cost-of-c}
\end{align}
On the other hand, the cost of $C$ w.r.t.\ $S_j$ is
\begin{align}
    \sum_{p\in S_j} \min_{c\in C_1}w(p)\|p-c\|^2 &= \sum_{p\in S_j} w(p)\|p-c_j\|^2 = \sum_{p\in S_j} w(p)\|p-jLe_0+ jLe_0-c_j\|^2 \nonumber\\
    &= \sum_{p\in S_j}w(p)\left(\|p-jLe_0\|^2 + 1 -2\langle p-jLe_0, jLe_0-c_j \rangle \right)\nonumber\\
    &= \sum_{p\in S_j} w(p)(\Delta_p^2+\|\Tilde{p}\|^2 +1) -2w(p)\langle p-jLe_0, jLe_0-c_j \rangle .\label{eqn:cost-of-c-to-S}
\end{align}
Recall $\Tilde{p}\in \R^d$ is $p_{1:d}$. 
For $j\in I$, the inner product is $0$, and thus the total cost w.r.t.\ $S$ is 
\begin{align*}
    \cost(S,C_1) = \sum_{p\in S} w(p)(\Delta_p^2+\|\Tilde{p}\|^2 +1) - 2\sum_{j\in \Bar{I}}\sum_{p\in S_j}w(p)\langle p-jLe_0, jLe_0-c_j \rangle,
\end{align*}
which finishes the proof.
\end{proof}

\noindent
For notational convenience, we define $\kappa := 2\sum_{j\in \Bar{I}}\sum_{p\in S_j}w(p)\langle p-jLe_0, jLe_0-c_j \rangle$.
Since $S$ is an $\eps$-coreset of $P$, we have 
\begin{align}\label{eqn:weight-constraints}
    dk/2- \eps dk/2\le \sum_{p\in S} w(p)(\Delta_p^2+\|p'\|^2 +1) - \kappa \le dk/2+\eps dk/2.
\end{align}

Next we consider a different set of $k$ centers denoted by $C_2$. By \Cref{lem:cost-to-smallset}, there exists unit vectors $v^j_1,v^j_2 \in \R^d$ such that
\begin{align}
    \sum_{p\in S_j} w(p)(\min_{\ell=1,2} \|\Tilde{p}-v^j_{\ell}\|^2+\Delta_p^2) \le \sum_{p\in S_j}w(p)(\|\Tilde{p}\|^2+1 + \Delta_p^2)-  \frac{2\sum_{p\in S_j} w(p)\|\Tilde{p}\|}{\sqrt{|S_j|}}. \label{eqn:small-coreset-cost}
\end{align}
Applying this to all $j\in I$ and get corresponding $v^j_1,v^j_2$ for all $j\in I$. Let $C_2=\{u_1^1,u_2^2,\cdots, u_1^{k/2},u_2^{k/2}\}$ be a set of $k$ centers in $\R^{d+1}$ defined as follows: if $j\in I$, $u_{\ell}^j$ is $v_{\ell}^j$ with an additional $0$th coordinate with value $jL$, making them lie in $H_j$; for $j\in \Bar{I}$, we use the same centers as in $C_1$, i.e., $u_{1}^j=u_{2}^j =c_j$. 

\begin{lemma}
    For $C_2$ constructed above, we have 
    \begin{align*}
        \cost(P,C_2) \ge \frac{kd}{2}-\sqrt{d}|I| \text{ and } 
    \end{align*} 
    \begin{align*}
        \cost(S,C_2) \le \sum_{p\in S}w(p)(\|\Tilde{p}\|^2+1+\Delta_p^2)-  \sum_{j\in I} \frac{2\sum_{p\in S_j} w(p)\|\Tilde{p}\|}{\sqrt{|S_j|}} -\kappa.
    \end{align*}
\end{lemma}
\begin{proof}
By \eqref{eqn:small-coreset-cost}, 
\begin{align*}
\cost(S,C_2) &= \sum_{j=1}^{k/2} \sum_{p\in S_j}w(p)\min_{c\in C_2}\|p-c\|^2 \\
&=  \sum_{j\in I}\sum_{p\in S_j} w(p)\min_{\ell=1,2} (\|\Tilde{p}-v^j_{\ell}\|^2+\Delta^2_p) +\sum_{j\in \Bar{I}} \sum_{p\in S_j}w(p)\|p-c_j\|^2\\
    &\le \sum_{p\in S}w(p)(\|\Tilde{p}\|^2+1+\Delta_p^2)-  \sum_{j\in I} \frac{2\sum_{p\in S_j} w(p)\|\Tilde{p}\|}{\sqrt{|S_j|}} -\kappa.
\end{align*}
By \Cref{lem:cost-to-basis} (with $k=2$), we have
\begin{align*}
    \sum_{i=1}^{d/2} \min_{\ell=1,2}\|e_i-v^j_{\ell}\|^2 \ge {d}-\sqrt{d}.
\end{align*}
It follows that
\begin{align*}
    \cost(P,C_2)&=\sum_{j=1}^{k/2}\sum_{i=1}^{d/2} \min_{c\in C_2}\|jLe_0+e_i-c\|^2 = \sum_{j\in I}\sum_{i=1}^{d/2} \min_{\ell=1,2} \|e_i-v^{j}_{\ell}\|^2 + \sum_{j\in \Bar{I}}\sum_{i=1}^{d/2} \|jLe_0+e_i-c\|^2 \\
    &\ge \frac{kd}{2}-\sqrt{d}|I|,
\end{align*}
where in the inequality, we also used the orthogonality between $e_i$ and $c_j-jLe_0$.
\end{proof}

\noindent
Since $S$ is an $\eps$-coreset of $P$, we have
\begin{align*}
    \frac{dk}{2}-|I|\sqrt{d} - \frac{\eps dk}{2} \le (\frac{dk}{2}-|I|\sqrt{d})(1-\eps) \le \sum_{p\in S}w(p)(\|\Tilde{p}\|^2+1+\Delta_p^2)-  \sum_{j\in I} \frac{2\sum_{p\in S_j} w(p)\|\Tilde{p}\|}{\sqrt{|S_j|}}-\kappa,
\end{align*}
which implies
\begin{align}
    \sum_{j\in I} \frac{2\sum_{p\in S_j} w(p)\|\Tilde{p}\|}{\sqrt{|S_j|}}
    &\le \sum_{p\in S}w(p)(\|\Tilde{p}\|^2+1+\Delta_p^2) -\frac{dk-2|I|\sqrt{d}-\eps kd}{2}-\kappa
    \nonumber\\
    &\le  \frac{dk+\eps dk}{2} - \frac{dk-2|I|\sqrt{d}-\eps kd}{2} \quad\textnormal{by \eqref{eqn:weight-constraints}} \nonumber\\
    &= |I|\sqrt{d} +\eps kd. \nonumber
\end{align}
By definition, $|S_j| \le d/4$, so
\begin{align*}
      \sum_{j\in I} \frac{2\sum_{p\in S_j} w(p)\|\Tilde{p}\|}{\sqrt{d/4}}\le \sum_{j\in I} \frac{2\sum_{p\in S_j} w(p)\|\Tilde{p}\|}{\sqrt{|S_j|}},
\end{align*}
and it follows that
\begin{align}\label{eqn:size-constraint}
     \frac{\sum_{j\in I} \sum_{p\in S_j} w(p)\|\Tilde{p}\|}{\sqrt{d}} 
     \le \frac{|I|\sqrt{d} +\eps kd} {4}.
\end{align}

Finally we consider a third set of $k$ centers $C_3$. Similarly, there are two centers per group. We set $m$ be a power of $2$ in $[d/2,d]$. Let $h_1,\cdots,h_m$ be the $m$-dimensional Hadamard basis vectors. So all $h_{\ell}$'s are $\{-\frac{1}{\sqrt{m}},\frac{1}{\sqrt{m}}\}$ vectors and $h_1=(\frac{1}{\sqrt{m}},\cdots,\frac{1}{\sqrt{m}})$. We slightly abuse notation and treat each $h_{\ell}$ as a $d$-dimensional vector by concatenating zeros in the end. For each $h_{\ell}$ construct a set of $k$ centers as follows. For each $j\in \Bar{I}$, we still use two copies of $c_j$.  For $j\in I$, the $0$th coordinate of the two centers is $jL$, then we concatenate $h_{\ell}$ and $-h_{\ell}$ respectively to the first and the second centers. 
\begin{lemma}\label{lem:hadamard-cost}
     Suppose $C_3$ is constructed based on $h_{\ell}$. Then for all $\ell\in [m]$, we have 
    \begin{align*}
        \cost(P,C_3) = \frac{kd}{2} -\frac{d|I|}{\sqrt{m}}\text{ and } 
    \end{align*} 
    \begin{align*}
        \cost(S,C_3) = \sum_{p\in S} w(p) (\|\Tilde{p}\|^2+1 + \Delta_p^2) - 2\sum_{j\in I} \sum_{p\in S_j}\langle w(p) \Tilde{p}, h^p_{\ell} \rangle -\kappa.
    \end{align*}
\end{lemma}
\begin{proof}
For $j\in I$, the cost of the two centers w.r.t.\ $P_j$ is
\begin{align*}
    \cost(P_j,C_3) = \sum_{i=1}^{d/2} \min_{s=-1,+1}\|e_i - s\cdot h_{\ell}\|^2 = \sum_{i=1}^{d/2} (2-2\max_{s=-1,+1}\langle h_{\ell},e_i\rangle)=\sum_{i=1}^{d/2} (2-\frac{2}{\sqrt{m}}) = d-\frac{d}{\sqrt{m}}.
\end{align*}
For $j\in \Bar{I}$, the cost w.r.t.\ $P_j$ is $d$ by \eqref{eqn:cost-of-c}.
Thus, the total cost over all subspaces is 
\begin{align*}
\cost(P,C_3) =    (d-\frac{d}{\sqrt{m}})|I| + \left(\frac{k}{2} -|I| \right)d = \frac{kd}{2} -\frac{d|I|}{\sqrt{m}}.
\end{align*} 

On the other hand, for $j\in I$, the cost w.r.t.\ $S_j$ is 
\begin{align*}
    \sum_{p\in S_j} w(p)(\Delta_p^2+ \min_{s=\{-1,+1\}} \|\Tilde{p}-s\cdot h_{\ell}\|^2) &=\sum_{p\in S_j} w(p) (\|\Tilde{p}\|^2+1 + \Delta_p^2 - 2\max_{s=\{-1,+1\}}\langle \Tilde{p}, s\cdot h_{\ell} \rangle)\\
   & =\sum_{p\in S_j} w(p) (\|\Tilde{p}\|^2+1 + \Delta_p^2 - 2\langle \Tilde{p}, h^p_{\ell} \rangle).
\end{align*}
Here $h^p_{\ell} = s^p\cdot h_{\ell}$, where  $s^p=\arg\max_{s=\{-1,+1\}}\langle \Tilde{p}, s\cdot h_{\ell} \rangle$.
For $j\in \Bar{I}$, the cost w.r.t.\ $S_j$ is $\sum_{p\in S_j} w(p)(\Delta_p^2+\|\Tilde{p}\|^2 +1) -2\langle p-jLe_0, jLe_0-c_j \rangle )$ by \eqref{eqn:cost-of-c-to-S}. Thus, the total cost w.r.t.\ $S$ is
\begin{align*}
    \cost(S,C_3) = \sum_{p\in S} w(p) (\|\Tilde{p}\|^2+1 + \Delta_p^2) - 2\sum_{j\in I} \sum_{p\in S_j}\langle w(p) \Tilde{p}, h^p_{\ell} \rangle -\kappa .
\end{align*}
This finishes the proof.
\end{proof}

\begin{corollary}
Let $S$ be a $\eps$-coreset of $P$, and $I =\{j: |S_j|\le d/4\}$. Then 
\begin{align*}
    \sum_{j\in I}\sum_{p\in S_j} w(p)\|\Tilde{p}\| \ge \frac{d|I|-\eps kd\sqrt{d}}{2}. \\
\end{align*}
\end{corollary}
\begin{proof}
Since $S$ is an $\eps$-coreset, we have by \Cref{lem:hadamard-cost}
\begin{align*}
    2\sum_{j\in I} \sum_{p\in S_j}\langle w(p) \Tilde{p}, h^p_{\ell} \rangle &\ge \sum_{p\in S} w(p) (\|\Tilde{p}\|^2+1 + \Delta_p^2) -\kappa -(\frac{kd}{2} -\frac{d|I|}{\sqrt{m}})(1+\eps)\\
    &\ge \sum_{p\in S} w(p) (\|\Tilde{p}\|^2+1 + \Delta_p^2) -\kappa -\frac{kd}{2} +\frac{d|I|}{\sqrt{m}}-\frac{\eps kd}{2}\\
    &\ge \frac{dk-\eps dk}{2} - \frac{kd}{2} +\frac{d|I|}{\sqrt{m}}-\frac{\eps kd}{2} \quad\textnormal{by \eqref{eqn:weight-constraints}}\\
    &= \frac{d|I|}{\sqrt{m}} - \eps kd.
\end{align*}
Note that the above inequality holds for all $\ell\in[m]$, then 
$$ 2\sum_{\ell=1}^m\sum_{j\in I} \sum_{p\in S_j}\langle w(p) \Tilde{p}, h^p_{\ell}\rangle\ge d|I|\sqrt{m} - \eps kdm.$$

By the Cauchy-Schwartz inequality, 
\begin{align*}
    \sum_{\ell=1}^m\sum_{j\in I} \sum_{p\in S_j}\langle w(p) \Tilde{p}, h^p_{\ell} \rangle &=  \sum_{j\in I}\sum_{p\in S_j}\langle w(p) \Tilde{p}, \sum_{\ell=1}^m h^p_{\ell} \rangle \\
    &\le \sum_{j\in I}\sum_{p\in S_j} w(p)\|\Tilde{p}\| \|\sum_{\ell=1}^m h^p_{\ell} \| \\
    &= \sqrt{m}\sum_{j\in I}\sum_{p\in S_j} w(p)\|\Tilde{p}\|.
\end{align*}
Therefore, we have
\begin{align*}
    \sum_{j\in I}\sum_{p\in S_j} w(p)\|\Tilde{p}\| \ge \frac{d|I|-\eps kd\sqrt{m}}{2} \ge \frac{d|I|-\eps kd\sqrt{d}}{2}.
\end{align*}
\end{proof}

\noindent
Combining the above corollary with \eqref{eqn:size-constraint}, we have
\begin{align*}
    \frac{\sqrt{d}|I| -\eps kd }{2} \le \frac{|I|\sqrt{d} + \eps kd}{4} \implies |I| \le 3\eps k\sqrt{d}.
\end{align*}
By the assumption $d\le \frac{1}{100\eps^2}$, it holds that $|I| \le \frac{3k}{10}$ or $|\Bar{I}|\ge \frac{k}{2} -\frac{3k}{10}=\frac{k}{5}$. Moreover, since $|S_j|>\frac{d}{4}$ for each $j\in \Bar{I}$, we have $|S|>\frac{d}{4}\cdot \frac{k}{5} = \frac{kd}{20}$.
\end{proof}


\section{Conclusion}

This work studies coresets for \kMedian problem in small dimensional Euclidean spaces. We give tight size bounds for \kMedian in $\R$ and show that the framework in \citep{HarPeled2005SmallerCF}, with significant improvement, is optimal.  For $d \ge 2$, we improve existing coreset upper bounds for \oneMedian and prove new lower bounds.

Our work leaves several interesting problems for future research. One of which is to close the gap between upper bounds and lower bounds for $d \ge 2$. Another one is to generalize our results to \kzC for general $z$. Note that the generalization is non-trivial even for $d=1$ since the cost function is piece-wise linear for \kMedian while piece-wise polynomial of order $z$ for general \kzC.

\newpage

\bibliographystyle{plainnat}
\bibliography{references}
	
\newpage
\appendix
\appendixpage
\addappheadtotoc

\section{Coreset Lower Bound for General \kMedian in $\mathbb{R}$}
\label{sec:general1dkmedianlower}
We prove the general case of \Cref{thm:Lowerbound1dkmedian} here.
\begin{proof}[the general case of \Cref{thm:Lowerbound1dkmedian}]

We first construct the hard instance $P$. Let $P_1$ denote the hard instance we have constructed in the proof of \Cref{thm:Lowerbound1dkmedian}. We take a large enough constant $L>0$, take $P_i = (i-1)L + P_1$, and take $P = \cup_{i=1}^{\frac{k}{2}} P_i$. Here $(i-1)L + P_1$ means $\{(i-1)L+p|p \in P_1\}$.

The dataset $P$ is a unification of $\frac{k}{2}$ copies of $P_1$. These copies are far from each other. Thus \kMedian problem on $P$ can be decomposed to  \twoMedian problem on each copy. We prove the \kMedian lower bound by applying the argument for the \twoMedian lower bound on every single copy and combining them together.

 We denote $P_1 = \cup_{j =1}^{\frac{1}{\epsilon}} I_{1,j}$, where $I_{1,j}$ is the $j$-th interval we constructed in the proof of the \twoMedian case of \Cref{thm:Lowerbound1dkmedian}. We denote $I_{i,j} = (i-1)L + I_{1,j}$, denote the left endpoint and right endpoint of $I_{i,j}$ by $l_{i,j}$ and $r_{i,j}$ respectively. We have $P_i = \cup_{j =1}^{\frac{1}{\epsilon}} I_{i,j}$.

Now, assume that $S$ is an $\frac{\epsilon}{300}$ coreset of $P$ such that $|S| < \frac{k}{4\epsilon}$. We prove that there must be a contradiction. Since $|S| < \frac{k}{4\epsilon}$, there must be at least half of $i$ such that $(l_{i,j_i}, r_{i,j_i}) \cap S = \varnothing$ for some $j_i$. We assume that these indexes are $1,2,\dots,\frac{k}{4}$, without loss of generality. We define a parametrized query family as $Q(t) = \cup_{i=1}^{\frac{k}{2}} Q_i(t)$, where $t \in [\frac{1}{3},1]$ and
\[
Q_i(t) = \begin{cases}
    \{l_{i,1}, l_{i,j_i} + t(r_{i,j_i} - l_{i,j_i}), r_{i,j_i} \} \quad \text{for } i \le \frac{k}{4},\\
    \{l_{i,1}\} \quad \text{otherwise}.
\end{cases}
\]

Consider $\cost(P,Q(t))$, a function of $t$. Since $L$ is large enough, we have $\cost(P,Q(t)) = \sum_{i=1}^{\frac{k}{2}} \cost(P_i,Q_i(t))$. The computation we have done in the proof of the \twoMedian case of \Cref{thm:Lowerbound1dkmedian} implies that $\cost(P_i,Q_i(t)) \le \frac{2}{\epsilon}$ for each $i$ and
\[(1-\frac{1}{3})^2 \frac{\mathrm{d^2}}{\mathrm{dt^2}} \cost(P_i,Q_i(t)) = 
\begin{cases}
    \frac{4}{9} \quad \text{for } i \le \frac{k}{4}, \\
    0 \quad \text{otherwise}.
\end{cases}
\]
Thus we have $\cost(P,Q(t)) \le \frac{k}{\epsilon}$ and $(1 - \frac{1}{3})^2 \frac{\mathrm{d^2}}{\mathrm{dt^2}} \cost(P,Q(t)) = \frac{k}{9}$.

It's easy to see that $\cost(S,Q(t))$ is affine linear since $(l_{i,j_i}, r_{i,j_i}) \cap S = \varnothing$ for $i \le \frac{k}{4}$. Since $S$ is an $\frac{\epsilon}{300}$ coreset, we have $|\cost(S,Q(t)) - \cost(P,Q(t))| \le \frac{\epsilon}{300} \cost(P,Q(t))$. By \Cref{lem:quadratic_approximation}, we must have $\frac{\epsilon}{300} \ge \frac{1}{32} \frac{\epsilon}{k} \frac{k}{9} > \frac{\epsilon}{300}$, which leads to a contradiction.
\end{proof}

\section{Proof of \Cref{thm:lowerbound} for General $z\geq 1$}
\label{sec:generalz}

Using similar ideas from \citep{cohenaddad2022towards}, our proof of the lower bound for $z=2$ can be extended to arbitrary $z$. First, we provide two lemmas analogous to \Cref{lem:cost-to-basis} and \Cref{lem:cost-to-smallset} for general $z\ge 1$. Their proofs can be found in Appendix A in \citep{cohenaddad2022towards}. 

\begin{lemma}\label{lem:cost-to-basis-z}
For any even number $k\ge 1$, let $\{c_1,\cdots,c_k\}$ be arbitrary $k$ unit vectors in $\R^d$ such that for each $i$ there exist some $j$ satisfying $c_i=-c_j$. We have
\begin{align*}
    \sum_{i=1}^{d/2} \min_{\ell=1}^k\|e_i-c_{\ell}\|^z \ge 2^{z/2-1}d - 2^{z/2}\max\{1, z/2\}\sqrt{\frac{kd}{2}}.
\end{align*}
\end{lemma}
\begin{lemma}\label{lem:cost-to-smallset-z}
Let $S$ be a set of points in $\R^d$ of size $t$ and $w: S\rightarrow \R^+$ be their weights. For arbitrary $\Delta_p$ for each $p$, there exist $2$ unit vectors $v_1, v_2$ satisfying $v_1=-v_2$, such that
\begin{align*}
    \sum_{p\in S} w(p)\min_{\ell=1,2} \left(\|p-v_{\ell}\|^2+\Delta_p^2 \right)^{z/2} \le &\sum_{s\in P}w(p)(\|p\|^2+1+\Delta_p^2)^{z/2}\\
    &-\min\{1,z/2\}\cdot \frac{2\sum_{p\in S} w(p)(\|p\|^2+1+\Delta_p^2)^{z/2-1} \|p\|}{\sqrt{t}}.
\end{align*}
\end{lemma}

\noindent
In this proof, the original point set $P$ and three sets of $k$-centers, namely $C_1,C_2,C_3$, are the same as for the case $z=2$. The difference is that now $I=\{j:|S_j|\le \frac{d}{2^z}\}$ and when constructing $C_2$, we use \Cref{lem:cost-to-smallset-z} in place of \Cref{lem:cost-to-smallset}. Again, we compare the cost of $P$ and $S$ w.r.t. $C_1,C_2,C_3$ and get the following lemmas.

\begin{lemma}
    For $C_1$ constructed above, we have $\cost(P,C_1) = \frac{kd}{4}\cdot 2^{z/2}$ and 
    \begin{align*}
        \cost(S,C_1) =\sum_{j\in {I}}\sum_{p\in S_j} w(p)(\Delta_p^2+\|\Tilde{p}\|^2 +1)^{z/2} +\sum_{j\in \Bar{I}}\sum_{p\in S_j}w(p)\|p-c_j\|^{z}.
    \end{align*}
\end{lemma}
\begin{proof}
Since $e_i$ is orthogonal to $c_j-jLe_0$ and $c_j-jLe_0$ has unit norm for all $i,j$, it follows that
\begin{align}
\cost(P,C_1) &=
    \sum_{j=1}^{k/2}\sum_{i=1}^{d/2} \min_{c\in C_1}\|jLe_0+e_i-c\|^{2\cdot z/2} = \sum_{j=1}^{k/2} \sum_{i=1}^{d/2} \|jLe_0+e_i -c_j\|^{2\cdot z/2} \nonumber\\
    &=\sum_{j=1}^{k/2}\sum_{i=1}^{d/2}(\|e_i\|^2+\|c_j-jLe_0\|^2-2\langle e_i, c_j-jLe_0 \rangle )^{z/2}\nonumber\\
    &=\frac{kd}{4}\cdot 2^{z/2}.\label{eqn:cost-of-c-z}
\end{align}
On the other hand, the cost of $C_1$ w.r.t.\ $S_j$ is
\begin{align}
    \sum_{p\in S_j} \min_{c\in C_1}w(p)\|p-c\|^{2\cdot z/2} &= \sum_{p\in S_j} w(p)\|p-c_j\|^{2\cdot z/2} = \sum_{p\in S_j} w(p)\|p-jLe_0+ jLe_0-c_j\|^{2\cdot z/2} \nonumber\\
    &= \sum_{p\in S_j}w(p) \left( \|p-jLe_0\|^2 + 1 -2\langle p-jLe_0, jLe_0-c_j \rangle \right)^{z/2}.\label{eqn:cost-of-c-to-S-z}
\end{align}
For $j\in I$, the inner product is $0$, and thus the total cost w.r.t.\ $S$ is 
\begin{align*}
    \cost(S,C_1) = \sum_{j\in {I}}\sum_{p\in S_j} w(p)(\Delta_p^2+\|\Tilde{p}\|^2 +1)^{z/2} + \sum_{j\in \Bar{I}}\sum_{p\in S_j}w(p)\|p-c_j\|^{z},
\end{align*}
which finishes the proof.
\end{proof}

\noindent
For notational convenience, we define $\kappa := \sum_{j\in \Bar{I}}\sum_{p\in S_j}w(p)\|p-c_j\|^{z}$.
Since $S$ is an $\eps$-coreset of $P$, we have 
\begin{align}\label{eqn:weight-constraints-z}
    \frac{kd}{4}\cdot 2^{z/2} - \frac{\eps kd}{4}\cdot 2^{z/2}  \le \sum_{j\in {I}}\sum_{p\in S_j} w(p)(\Delta_p^2+\|\Tilde{p}\|^2 +1)^{z/2} + \kappa \le \frac{kd}{4}\cdot 2^{z/2} + \frac{\eps kd}{4}2^{z/2}.
\end{align}

Next we consider a different set of $k$ centers denoted by $C_2$. By \Cref{lem:cost-to-smallset-z}, there exists unit vectors $v^j_1,v^j_2 \in \R^d$ satisfying $v^j_1=-v^j_2$ such that
\begin{align}
    \sum_{p\in S_j} w(p)(\min_{\ell=1,2} \left(\|\Tilde{p}-v^j_{\ell}\|^2+\Delta^2_p\right)^{z/2}) \le&
    \sum_{p\in S_j}w(p)(\|\Tilde{p}\|^2+1 +\Delta^2_p )^{z/2} \nonumber\\
    &-  \min\{1,z/2\}\frac{2\sum_{p\in S_j} w(p)(\|\Tilde{p}\|^2+1 +\Delta^2_p)^{z/2-1} \|\Tilde{p}\|}{\sqrt{|S_j|}}. \label{eqn:small-coreset-cost-z}
\end{align}
Applying this to all $j\in I$ and get corresponding $v^j_1,v^j_2$ for all $j\in I$. Let $C_2=\{u_1^1,u_2^2,\cdots, u_1^{k/2},u_2^{k/2}\}$ be a set of $k$ centers in $\R^{d+1}$ defined as follows: if $j\in I$, $u_{\ell}^j$ is $v_{\ell}^j$ with an additional $0$th coordinate with value $jL$, making them lie in $H_j$; for $j\in \Bar{I}$, we use the same centers as in $C_1$, i.e., $u_{1}^j=u_{2}^j =c_j$.

\begin{lemma}
    For $C_2$ constructed above, we have 
    \begin{align*}
        \cost(P,C_2) \ge 2^{z/2}\left(\frac{kd}{4} -\max\{1, z/2\}\sqrt{d}|I|\right), \text{ and } 
    \end{align*} 
    \begin{align*}
        \cost(S,C_2) \le& \sum_{j\in I}\sum_{p\in S_j}w(p)(\|\Tilde{p}\|^2+1 +\Delta^2_p )^{z/2} \\
    &-  \min\{1,z/2\}\sum_{j\in I}\frac{2\sum_{p\in S_j} w(p)(\|\Tilde{p}\|^2+1 +\Delta^2_p)^{z/2-1} \|\Tilde{p}\|}{\sqrt{|S_j|}}+\kappa.
    \end{align*}
\end{lemma}
\begin{proof}
By \eqref{eqn:small-coreset-cost-z}, 
\begin{align*}
\cost(S,C_2) &= \sum_{j=1}^{k/2} \sum_{p\in S_j}w(p)\min_{c\in C_2}\|p-c\|^{2\cdot z/2} =
    \sum_{j\in I}\sum_{p\in S_j} w(p)\min_{\ell=1,2} (\|\Tilde{p}-v^j_{\ell}\|^2+\Delta^2_p)^{z/2} +\kappa\\
    &\le 
    \sum_{j\in I}\sum_{p\in S_j}w(p)(\|\Tilde{p}\|^2+1 +\Delta^2_p )^{z/2} \\
    &-  \min\{1,z/2\}\sum_{j\in I}\frac{2\sum_{p\in S_j} w(p)(\|\Tilde{p}\|^2+1 +\Delta^2_p)^{z/2-1} \|\Tilde{p}\|}{\sqrt{|S_j|}}+\kappa.
\end{align*}
By \Cref{lem:cost-to-basis-z} (with $k=2$), we have
\begin{align*}
    \sum_{i=1}^{d/2} \min_{\ell=1,2}\|e_i-v^j_{\ell}\|^z \ge 2^{z/2-1}d - 2^{z/2}\max\{1, z/2\}\sqrt{d}.
\end{align*}
It follows that
\begin{align*}
    \cost(P,C_2)&=\sum_{j=1}^{k/2}\sum_{i=1}^{d/2} \min_{c\in C_2}\|jLe_0+e_i-c\|^{z}\\
    &= \sum_{j\in I}\sum_{i=1}^{d/2} \min_{\ell=1,2} \|e_i-v^{j}_{\ell}\|^{2\cdot z/2} + \sum_{j\in \Bar{I}}\sum_{i=1}^{d/2} \|jLe_0+e_i-c_j\|^{2\cdot z/2} \\
    & \ge \left( 2^{z/2-1}d - 2^{z/2}\max\{1, z/2\}\sqrt{d} \right)|I| + |\Bar{I}|\frac{d}{2}\cdot 2^{z/2}\\
    &= \frac{kd}{4}2^{z/2} - 2^{z/2}\max\{1, z/2\}\sqrt{d}|I|,
\end{align*}
where in the inequality, we also used the orthogonality between $e_i$ and $c_j-jLe_0$.
\end{proof}

\noindent
Since $S$ is an $\eps$-coreset of $P$, we have
\begin{align*}
    &2^{z/2}\left(\frac{dk}{4}-\max\{1, z/2\}|I|\sqrt{d} - \frac{\eps dk}{4}\right) \le 2^{z/2}\left(\frac{kd}{4} -\max\{1, z/2\}\sqrt{d}|I|\right)(1-\eps) \\
    &\le \sum_{j\in I}\sum_{p\in S_j}w(p)(\|\Tilde{p}\|^2+1 +\Delta^2_p )^{z/2} -  \min\{1,z/2\}\sum_{j\in I}\frac{2\sum_{p\in S_j} w(p)(\|\Tilde{p}\|^2+1 +\Delta^2_p)^{z/2-1} \|\Tilde{p}\|}{\sqrt{|S_j|}}+\kappa,
\end{align*}
which implies
\begin{align}
    &\min\{1,z/2\}\sum_{j\in I}\frac{2\sum_{p\in S_j} w(p)(\|\Tilde{p}\|^2+1 +\Delta^2_p)^{z/2-1} \|\Tilde{p}\|}{\sqrt{|S_j|}}\nonumber\\
    &\le \sum_{j\in I}\sum_{p\in S_j}w(p)(\|\Tilde{p}\|^2+1 +\Delta^2_p )^{z/2} - 2^{z/2}\left(\frac{dk}{4}-\max\{1, z/2\}|I|\sqrt{d} - \frac{\eps dk}{4}\right) + \kappa
    \nonumber\\
    &\le  \frac{kd}{4}\cdot 2^{z/2} + \frac{\eps kd}{4}2^{z/2}  - 2^{z/2}\left(\frac{dk}{4}-\max\{1, z/2\}|I|\sqrt{d} - \frac{\eps dk}{4}\right) \quad\textnormal{by \eqref{eqn:weight-constraints-z}} \nonumber\\
    &= \max\{1, z/2\} |I|\sqrt{d}2^{z/2} +\frac{\eps kd}{2}2^{z/2}. \nonumber
\end{align}
By definition, $|S_j| \le d/t^2$, so
\begin{align*}
      & \min\{1,\frac{z}{2}\} \sum_{j\in I} \frac{2\sum_{p\in S_j} w(p)(\|\Tilde{p}\|^2+1 +\Delta^2_p)^{z/2-1} \|\Tilde{p}\|}{\sqrt{d/t^2}}\\
      \le &\min\{1,\frac{z}{2}\}\sum_{j\in I}\frac{2\sum_{p\in S_j} w(p)(\|\Tilde{p}\|^2+1 +\Delta^2_p)^{z/2-1} \|\Tilde{p}\|}{\sqrt{|S_j|}},
\end{align*}
and it follows that
\begin{align}\label{eqn:size-constraint-z}
     \min\{1,\frac{z}{2}\} \sum_{j\in I} \frac{\sum_{p\in S_j} w(p)(\|\Tilde{p}\|^2+1 +\Delta^2_p)^{z/2-1} \|\Tilde{p}\|}{\sqrt{d}} 
     \le \frac{\max\{1, z/2\} |I|\sqrt{d} 2^{z/2} +\frac{\eps kd}{2}2^{z/2}} {2t}.
\end{align}

Finally we consider a third set of $k$ centers $C_3$. Similarly, there are two centers per group. We set $m$ be a power of $2$ in $[d/2,d]$. Let $h_1,\cdots,h_m$ be the $m$-dimensional Hadamard basis vectors. So all $h_{\ell}$'s are $\{-\frac{1}{\sqrt{m}},\frac{1}{\sqrt{m}}\}$ vectors and $h_1=(\frac{1}{\sqrt{m}},\cdots,\frac{1}{\sqrt{m}})$. We slightly abuse notation and treat each $h_{\ell}$ as a $d$-dimensional vector by concatenating zeros in the end. For each $h_{\ell}$ construct a set of $k$ centers as follows. For each $j\in \Bar{I}$, we still use two copies of $c_j$.  For $j\in I$, the $0$th coordinate of the two centers is $jL$, then we concatenate $h_{\ell}$ and $-h_{\ell}$ respectively to the first and the second centers.

\begin{lemma}\label{lem:hadamard-cost-z}
     Suppose $C_3$ is constructed based on $h_{\ell}$. Then for all $\ell\in [m]$, we have 
    \begin{align*}
        \cost(P,C_3) \le 2^{z/2}\left( \frac{kd}{4} - \frac{d|I|}{2}\cdot \frac{\min\{1,z/2\}}{\sqrt{m}} \right), \text{ and } 
    \end{align*} 
    \begin{align*}
        \cost(S,C_3) &\ge \sum_{j\in I}\sum_{p\in S_j} w(p)(\|\Tilde{p}\|^2+1 + \Delta_p^2)^{\frac{z}{2}} \\
        -& 2\max\{1,\frac{z}{2}\} \sum_{j\in I}\sum_{p\in S_j} w(p)\langle \Tilde{p}, h^p_{\ell} \rangle (\|\Tilde{p}\|^2+1 + \Delta_p^2)^{\frac{z}{2}-1}+\kappa.
    \end{align*}
\end{lemma}

\begin{proof}
For $j\in I$, the cost of the two centers w.r.t.\ $P_j$ is
\begin{align*}
    \cost(P_j,C_3) &= \sum_{i=1}^{d/2} \min_{s=-1,+1}\|e_i - s\cdot h_{\ell}\|^z = \sum_{i=1}^{d/2} (2-2\max_{s=-1,+1}\langle h_{\ell},e_i\rangle)^{z/2}= \frac{d}{2}(2-\frac{2}{\sqrt{m}})^{z/2} \\
    &\le \frac{d}{2}\cdot 2^{z/2} \left(1-\frac{\min\{1,z/2\}}{\sqrt{m}}\right).
\end{align*}
For $j\in \Bar{I}$, the cost w.r.t.\ $P_j$ is $\frac{d}{2}\cdot 2^{z/2}$ by \eqref{eqn:cost-of-c-z}.
Thus, the total cost over all subspaces is 
\begin{align*}
\cost(P,C_3) &\le    \frac{d}{2}\cdot2^{z/2} \left(1-\frac{\min\{1,z/2\}}{\sqrt{m}}\right)|I| + \left(\frac{k}{2} -|I| \right)\frac{d}{2}\cdot 2^{z/2}\\ 
&=  2^{z/2}\left( \frac{kd}{4} - \frac{d|I|}{2}\cdot \frac{\min\{1,z/2\}}{\sqrt{m}} \right).
\end{align*} 

On the other hand, for $j\in I$, the cost w.r.t.\ $S_j$ is 
\begin{align*}
     &\sum_{p\in S_j} w(p)(\Delta_p^2+ \min_{s=\{-1,+1\}} \|\Tilde{p}-s\cdot h_{\ell}\|^2)^{z/2}\\
    =&\sum_{p\in S_j} w(p) (\|\Tilde{p}\|^2+1 + \Delta_p^2 - 2\max_{s=\{-1,+1\}}\langle \Tilde{p}, s\cdot h_{\ell} \rangle)^{z/2}\\
    =& \sum_{p\in S_j} w(p) (\|\Tilde{p}\|^2+1 + \Delta_p^2 - 2\langle \Tilde{p}, h^p_{\ell} \rangle)^{z/2}\\
   \ge& \sum_{p\in S_j} w(p)(\|\Tilde{p}\|^2+1 + \Delta_p^2)^{\frac{z}{2}} - 2\max\{1,\frac{z}{2}\} \sum_{p\in S_j} w(p)\langle \Tilde{p}, h^p_{\ell} \rangle (\|\Tilde{p}\|^2+1 + \Delta_p^2)^{\frac{z}{2}-1}.
\end{align*}
Here $h^p_{\ell} = s^p\cdot h_{\ell}$, where  $s^p=\arg\max_{s=\{-1,+1\}}\langle \Tilde{p}, s\cdot h_{\ell} \rangle$.
For $j\in \Bar{I}$, the total cost w.r.t.\ $S_j$ is $\kappa$. Thus, the total cost w.r.t.\ $S$ is
\begin{align*}
    \cost(S,C_3) &\ge \sum_{j\in I}\sum_{p\in S_j} w(p)(\|\Tilde{p}\|^2+1 + \Delta_p^2)^{\frac{z}{2}}\\
    -& 2\max\{1,\frac{z}{2}\} \sum_{j\in I}\sum_{p\in S_j} w(p)\langle \Tilde{p}, h^p_{\ell} \rangle (\|\Tilde{p}\|^2+1 + \Delta_p^2)^{\frac{z}{2}-1}+\kappa .
\end{align*}
This finishes the proof.
\end{proof}

\begin{corollary}
Let $S$ be a $\eps$-coreset of $P$, and $I =\{j: |S_j|\le d/4\}$. Then 
\begin{align*}
    2\max\{1,\frac{z}{2}\} \sum_{j\in I}\sum_{p\in S_j} w(p)(\|\Tilde{p}\|^2+1 + \Delta_p^2)^{\frac{z}{2}-1} \|\Tilde{p}\| \ge 2^{z/2}\cdot \left( \frac{d|I|}{2}\cdot\min\{1,z/2\} - \frac{\eps kd \sqrt{d}}{2}\right). \\
\end{align*}
\end{corollary}
\begin{proof}
Since $S$ is an $\eps$-coreset, we have by \Cref{lem:hadamard-cost-z}
\begin{align*}
    &2\max\{1,\frac{z}{2}\} \sum_{j\in I}\sum_{p\in S_j} w(p)\langle \Tilde{p}, h^p_{\ell} \rangle (\|\Tilde{p}\|^2+1 + \Delta_p^2)^{\frac{z}{2}-1}\\
    & \ge \sum_{j\in I}\sum_{p\in S_j} w(p)(\|\Tilde{p}\|^2+1 + \Delta_p^2)^{\frac{z}{2}} + \kappa 
    -2^{z/2}\left( \frac{kd}{4} - \frac{d|I|}{2}\cdot \frac{\min\{1,z/2\}}{\sqrt{m}} \right)(1+\eps)\\
    &\ge \frac{kd}{4}\cdot 2^{z/2} - \frac{\eps kd}{4}\cdot 2^{z/2} - 2^{z/2}\left( \frac{kd}{4} - \frac{d|I|}{2}\cdot \frac{\min\{1,z/2\}}{\sqrt{m}} +\frac{\eps kd}{4} \right) \quad\textnormal{by \eqref{eqn:weight-constraints-z}}\\
    &= 2^{z/2}\cdot \frac{d|I|}{2}\cdot \frac{\min\{1,z/2\}}{\sqrt{m}} - \frac{\eps kd}{2}\cdot 2^{z/2}.
\end{align*}
Note that the above inequality holds for all $\ell\in[m]$, then 
$$ 2\max\{1,\frac{z}{2}\} \sum_{\ell=1}^m\sum_{j\in I}\sum_{p\in S_j} w(p)\langle \Tilde{p}, h^p_{\ell} \rangle (\|\Tilde{p}\|^2+1 + \Delta_p^2)^{\frac{z}{2}-1} \ge 2^{z/2}\cdot \left( \frac{d|I|\sqrt{m}}{2}\cdot\min\{1,z/2\} - \frac{\eps kd m}{2}\right).$$

By the Cauchy-Schwartz inequality, 
\begin{align*}
    \sum_{\ell=1}^m\sum_{j\in I}\sum_{p\in S_j} w(p)\langle \Tilde{p}, h^p_{\ell} \rangle (\|\Tilde{p}\|^2+1 + \Delta_p^2)^{\frac{z}{2}-1} &=  \sum_{j\in I}\sum_{p\in S_j} w(p)\langle \Tilde{p}, \sum_{\ell=1}^m h^p_{\ell} \rangle (\|\Tilde{p}\|^2+1 + \Delta_p^2)^{\frac{z}{2}-1}  \\
    &\le \sum_{j\in I}\sum_{p\in S_j} w(p)(\|\Tilde{p}\|^2+1 + \Delta_p^2)^{\frac{z}{2}-1} \|\Tilde{p}\|\cdot  \|\sum_{\ell=1}^m h^p_{\ell}\|   \\
    &= \sqrt{m}\sum_{j\in I}\sum_{p\in S_j} w(p)(\|\Tilde{p}\|^2+1 + \Delta_p^2)^{\frac{z}{2}-1} \|\Tilde{p}\|.
\end{align*}
Therefore, we have
\begin{align*}
    2\max\{1,\frac{z}{2}\} \sum_{j\in I}\sum_{p\in S_j} w(p)(\|\Tilde{p}\|^2+1 + \Delta_p^2)^{\frac{z}{2}-1} \|\Tilde{p}\| &\ge  2^{z/2}\cdot \left( \frac{d|I|}{2}\cdot\min\{1,z/2\} - \frac{\eps kd \sqrt{m}}{2}\right) \\
    &\ge 2^{z/2}\cdot \left( \frac{d|I|}{2}\cdot\min\{1,z/2\} - \frac{\eps kd \sqrt{d}}{2}\right).
\end{align*}
\end{proof}

Combining the above corollary with \eqref{eqn:size-constraint-z}, we have
\begin{align*}
    \frac{\min\{1,z/2\}}{2\max\{1,z/2\}}2^{z/2}\cdot \left( \frac{\sqrt{d}|I|}{2}\cdot\min\{1,z/2\} - \frac{\eps kd }{2}\right) \le \frac{\left( \max\{1, z/2\} |I|\sqrt{d} +\frac{\eps kd}{2}\right)2^{z/2}} {2t},
\end{align*}
which implies that
\begin{align*}
    \left(\frac{\min\{1,(z/2)^2\}}{4\max\{1,(z/2)\}}-\frac{\max\{1,z/2\}}{2t}\right)|I| \le  \frac{\min\{1,(z/2)\}\eps kd}{4\max\{1,(z/2)\}}+ \frac{\eps k\sqrt{d}}{4t}.
\end{align*}
So if we set $t = \frac{4\max\{1,(z/2)^2\}}{\min\{1,(z/2)^2\}}$, then
\begin{align*}
    \frac{\min\{1,(z/2)^2\}}{8\max\{1,(z/2)\}}|I| \le \frac{\min\{1,(z/2)\}\eps k\sqrt{d}}{2\max\{1,(z/2)\}} \implies |I| \le \frac{4\eps k\sqrt{d}}{\min\{1,z/2\}}.
\end{align*}

By the assumption $d\le \frac{\min\{1,(z/2)^2\}}{100\eps^2}$, it holds that $|I| \le \frac{2k}{5}$ or $|\Bar{I}|\ge \frac{k}{2} -\frac{2k}{5}=\frac{k}{10}$. Moreover, since $|S_j|>\frac{d}{t^2}$ for each $j\in \Bar{I}$, we have $|S|>\frac{d}{t^2}\cdot \frac{k}{5} = \frac{kd \min\{1,(z/2)^4\}}{\max\{1,(z/2)^4\}}$.

\end{document}